\documentclass[10pt]{article} 
\usepackage[preprint]{tmlr}

\usepackage{amsmath,amsfonts,bm}

\def\eqref#1{equation~\ref{#1}}

\def\1{\bm{1}}

\def\ra{{\textnormal{a}}}

\def\rx{{\textnormal{x}}}

\def\rva{{\mathbf{a}}}

\def\erva{{\textnormal{a}}}

\def\ervx{{\textnormal{x}}}

\def\rmA{{\mathbf{A}}}

\def\vmu{{\bm{\mu}}}
\def\vtheta{{\bm{\theta}}}
\def\va{{\bm{a}}}

\def\ve{{\bm{e}}}

\def\vx{{\bm{x}}}

\def\eva{{a}}

\def\mA{{\bm{A}}}

\def\mH{{\bm{H}}}
\def\mI{{\bm{I}}}
\def\mJ{{\bm{J}}}

\def\mX{{\bm{X}}}

\def\mSigma{{\bm{\Sigma}}}

\DeclareMathAlphabet{\mathsfit}{\encodingdefault}{\sfdefault}{m}{sl}
\SetMathAlphabet{\mathsfit}{bold}{\encodingdefault}{\sfdefault}{bx}{n}
\newcommand{\tens}[1]{\bm{\mathsfit{#1}}}
\def\tA{{\tens{A}}}

\def\tX{{\tens{X}}}

\def\gG{{\mathcal{G}}}

\def\sA{{\mathbb{A}}}
\def\sB{{\mathbb{B}}}

\def\sS{{\mathbb{S}}}

\def\emA{{A}}

\newcommand{\etens}[1]{\mathsfit{#1}}

\def\etA{{\etens{A}}}

\newcommand{\E}{\mathbb{E}}

\newcommand{\R}{\mathbb{R}}

\newcommand{\KL}{D_{\mathrm{KL}}}
\newcommand{\Var}{\mathrm{Var}}

\newcommand{\Cov}{\mathrm{Cov}}

\newcommand{\normltwo}{L^2}
\newcommand{\normlp}{L^p}

\newcommand{\parents}{Pa} 

\usepackage{hyperref}
\usepackage{url}

\title{Networked Communication for Decentralised Agents in Mean-Field Games}

\author{\name Patrick Benjamin \email patrick.benjamin@cs.ox.ac.uk \\
      \addr Department of Computer Science\\
      University of Oxford
      \AND
      \\\;\\
      \name Alessandro Abate \email alessandro.abate@cs.ox.ac.uk \\
      \addr Department of Computer Science\\
      University of Oxford
      }

\def\month{MM}  
\def\year{YYYY} 
\def\openreview{\url{https://openreview.net/forum?id=XXXX}} 

\usepackage{amsthm}

\theoremstyle{plain}
\newtheorem{theorem}{Theorem}[section]
\newtheorem{proposition}[theorem]{Proposition}
\newtheorem{lemma}[theorem]{Lemma}

\theoremstyle{definition}  
\newtheorem{definition}[theorem]{Definition}
\newtheorem{assumption}[theorem]{Assumption}
\newtheorem{remark}[theorem]{Remark}

\let\AND\relax

\usepackage{algorithm}
\usepackage{algorithmic}
\usepackage{graphicx}
\usepackage{booktabs}
\usepackage{dsfont}
\usepackage{subcaption}

\begin{document}

\maketitle

\begin{abstract}
Methods like multi-agent reinforcement learning struggle to scale with growing population size. Mean-field games (MFGs) are a game-theoretic approach that can circumvent this by finding a solution for an abstract infinite population, which can then be used as an approximate solution for the $N$-agent problem. However, classical mean-field algorithms usually only work under restrictive conditions. We take steps to address this by introducing networked communication to MFGs, in particular to settings that use a single, non-episodic run of $N$ decentralised agents to simulate the infinite population, as is likely to be most reasonable in real-world deployments.
We prove that our architecture's sample guarantees lie between those of earlier theoretical algorithms for the centralised- and independent-learning architectures, varying dependent on network structure and the number of communication rounds.
However, the sample guarantees of the three theoretical algorithms do not actually result in practical convergence times. We thus contribute practical enhancements to all three algorithms allowing us to present their first empirical demonstrations
. We then show that in practical settings where the theoretical hyperparameters are not observed, giving fewer loops but poorer estimation of the Q-function, our communication scheme still respects the earlier theoretical analysis: it considerably accelerates learning over the independent case, which hardly seems to learn at all, and 
often performs similarly to the centralised case, while removing the restrictive assumption of the latter. We provide ablations and additional studies showing that our networked approach also has advantages over both alternatives in terms of robustness to update failures and to changes in population size.
\end{abstract}

\section{Introduction}\label{introduction}

Multi-agent reinforcement learning (MARL) \citep{4445757, Zhang2021} generalises reinforcement learning (RL) \citep{sutton2018reinforcement} to sequential decision problems in multi-agent systems, that is, temporally extended interactions amongst multiple agents within a common environment. Underpinned by breakthroughs in deep learning, MARL has recently seen empirical success in a variety of domains, including robotics \citep{LEOTTAU2018130, 10107729, multi_robot_survey}, autonomous driving and infrastructure \citep{shalev2016safe,Mannion2016}, games \citep{10.5555/3306127.3332052,alphastarblog,openai2019dota}, social science / cooperative AI \citep{social_dilemmas, cao2018emergent, jaques2019social,mckee2020social} and economics \citep{iet:/content/journals/10.1049/iet-gtd.2016.0075,SHAVANDI2022118124}. However, it has usually been difficult to scale MARL systems beyond configurations with agents numbering in the low tens, as the joint state and action spaces grow exponentially with the number of agents \citep{nash_complexity,Vinyals2019GrandmasterLI, continuous_fictitious_play,NEURIPS2020_e9bcd1b0, one_that_may_sample,scalable_deep,SHAVANDI2022118124,yardim2024exploiting,zeng2024single,lirevisiting}, which represents a computational bottleneck in representing or even exploring them. Nevertheless, the value of reasoning about interactions among very large populations of agents has been recognised, and an informal distinction is sometimes drawn between multi- and \textit{many}-agent systems \citep{zheng2018magent,curse,cui2022survey}. The latter situation can be more useful (as in cases where better solutions arise from the presence of more agents \citep{shiri2019massive,daniel_ants,multi_robot_survey,10.1002/aaai.12131}), more parallelisable \citep{andreen2016emergent}, more fault-tolerant \citep{CHANG2023102570}, or otherwise more reflective of certain real-world systems involving large numbers of decision makers, such as financial markets \citep{iet:/content/journals/10.1049/iet-gtd.2016.0075, TRIMBORN2018613,SHAVANDI2022118124}, cryptocurrency mining \citep{li2022mean}, smart infrastructures with large populations of autonomous vehicles \citep{DBLP:journals/corr/abs-2001-03232,Kuang_Huang_2020,dynamic_traffic}, cloud resource management \citep{nips2022mao}, smart grids, and other large-scale cyber-physical systems \citep{10093067}. 

One approach that has been used to circumvent the scalability issue in MARL is the game-theoretic mean-field framework \citep{huangMFG,lasry2007mean}, inspired by theoretical physics.\footnote{Note that the mean-field framework has an independent scientific lineage to MARL and encompasses stochastic control, Partial Differential Equations and others. While RL can be used as a learning method within the mean-field framework (as in our case), other (model-based) methods have traditionally been used, and the MFG framework has considered structural questions regarding the existence, uniqueness and stability of equilibria even without running a learning algorithm.} This models a representative agent as interacting not with the other individuals in the population on a per-agent basis, but instead with a distribution over the other agents, known as the \textit{mean field}. Defined formally in Sec. \ref{preliminaries}, the framework analyses the limiting case when the population consists of an infinite number of symmetric and anonymous agents, that is, they have identical reward and transition functions which depend on the mean field rather than on the actions of specific other players. Considering strategic interactions in populations that are so large that the individual identities of agents are not significant (i.e. they are anonymous) is a modelling choice that fits many systems, such as crowd motion, traffic flows, epidemics, financial systems, 
electric vehicle charging and energy demand response. Anonymity facilitates the beneficial scaling properties of MFGs with respect to very large populations. It allows agents to be exchangeable and their strategies symmetrical, massively reducing complexity. This allows agents to interact via the aggregate distribution instead of via individual identities, which would hinder the simplification afforded by MFGs, and in turn their scalability. Variants have been developed beyond the classical symmetric/anonymous assumptions, including heterogeneous-type/multi-population mean-field problems, major-minor mean-field problems and graphon mean-field problems \citep{survey_learningMFGs}, though we consider the typical formulation in this work.

\begin{figure*}[t]
    \centering
    \begin{subfigure}[b]{0.9\columnwidth}
        \centering
        \includegraphics[width=\textwidth]{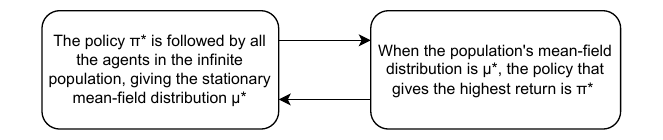}
        \caption{}
    \end{subfigure}
    
    \begin{subfigure}[b]{0.9\columnwidth}
        \centering
        \includegraphics[width=\textwidth]{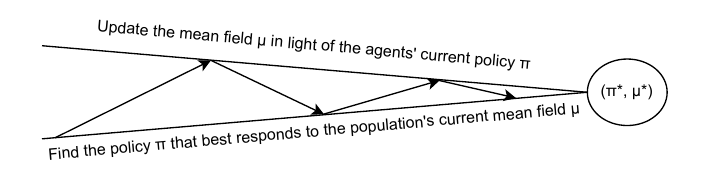}
        \caption{}
    \end{subfigure}
    \caption{a) As formalised in Def. \ref{Nash_equilibrium_of_stationary_MFG}, an optimal solution to a MFG is a policy that best responds to the stationary mean-field distribution that arises when all agents follow that same policy. b) The solution can therefore be seen as the fixed point of two operators: finding the mean-field distribution that arises when agents follow a given policy and finding the policy that best responds to a given mean field.}
    \label{MFG_NE_diagram}
\end{figure*}

In this work we focus on mean-field games (MFG), which are non-cooperative scenarios where each agent seeks to maximise its individual return; in particular we focus on MFGs with stationary population distributions (\textit{stationary MFGs}), where learning is more tractable than in non-stationary ones \citep{one_that_may_sample,Anahtarci2020QLearningIR,zaman2023oraclefree,policy_mirror_independent,li2025effect,osborne2025rates}. The solution to a stationary MFG is a MFG-Nash equilibrium (MFG-NE, see Fig. \ref{MFG_NE_diagram}), which reflects the situation when each agent responds optimally to the population distribution that arises when all other agents follow that same optimal behaviour.

The MFG-NE can be used as an approximation for the Nash equilibrium (NE) solution in the associated finite-agent game, which is computationally more difficult to solve in itself (hence the scalability issue in MARL), with the error in the solution reducing as the number of agents \textit{N} tends to infinity \citep{finite_proofs, doi:10.1137/20M1360700,10.5555/3586589.3586718,major_minor_control,cui2023learningdecentralisedcontrol,Anahtarci2020QLearningIR,yardim2024mean, TOUMI2024111420, hu2024mfoml,chen2024periodic,bayraktar2024learning, 10928821,yardim2025variational}
. For example, it might be very difficult to directly find a NE for a million agents. However this can be circumvented by finding the solution for the infinite population, where analysis is simplified by taking place in the mathematical limit of population size, and then applying that back to the million agents. Moreover, once the solution is found it does not depend on the size of the deployed population, so some of the million agents could leave without requiring the rest of the population to compute a new solution, or the population could grow to 10 million and the original solution would work even better than before. Furthermore, although the analysis assumes an infinite population, there are numerous ways to represent and generate its behaviour. While it might be calculated analytically or by extrapolating from the behaviour of a single generic agent that is assumed to represent the whole population, more recent work has involved simulating the infinite population by drawing finite random samples \citep{one_that_may_sample}, or indeed by deploying an empirical population consisting of a finite number of agents that is assumed to be representative of the infinite population \citep{policy_mirror_independent}. The latter case, i.e. using a finite population to simulate the infinite one, might be particularly desirable when the infinite limit is being used precisely to find an approximate solution for the NE of a finite-population problem - the same finite population that was of original interest might be able to be used find the mean-field solution that approximates its own solution. See Rem. \ref{two_conceptions} for further discussion.

Works have therefore considered applying MFGs to find approximate solutions for a wide variety of real-world problems involving a large but finite number of agents, which might otherwise have been too difficult to solve - note though that these are generally idealised simulated environments serving as proxies for real-world. 
Examples include: 
\begin{itemize}
    \item financial markets \citep{TRIMBORN2018613,BERNASCONI2023101974,FOGUENTCHUENDOM2024111878,chen2024periodic,becherer2024common,zhang2024mean,chen2024deciding,cecchin2025weak,bo2025mean,fu2025mean,wang2025primal, moll2025mean,tchuendom2025ranking,WANG2025112260,li_diffusion_2025}; ticket pricing \citep{aydin2025fare}; the green economy \citep{math13050691,aksamit2025switching}; electricity markets \citep{feng2025decentralized,he2025hybrid};
    \item autonomous vehicles \citep{Kuang_Huang_2020,mo2024game,pande2025generative,10456564,chen2024bridging}; traffic signal control \citep{dynamic_traffic}; ride-hailing platforms \citep{li2025repositioning}; electric vehicle charging \citep{10.1145/3576914.3587709,10741829,10907167};
    \item  cryptocurrency mining \citep{li2022mean,garcia2025mean}; edge computing \citep{aggarwal2024mean,10496451,app14093538,aggarwal2025distributed}; cloud resource management \citep{nips2022mao,Kang2025}; smart grids, and other large-scale cyber-physical systems \citep{7300401, 9894296,10093067,10186687,wang2024mean,wu2024hi,10636723,10928821}.
    \item swarms \citep{doi:10.1142/S0219198924400085,10919071}; defence \citep{aerospace12040302}; communication networks \citep{9299742,10508221,10637582,10809790,10689365,Kang2025satellite,choutri2025backpressure};  data collection by UAVs \citep{10508811}; 
    \item social network modelling \citep{Kang2025_Opinion_Evolution}; crowd modelling \citep{DU685643}; crowdsensing \citep{10143734};  
    \item pollution regulation \citep{del2024mean}; resource management in fisheries \citep{YOSHIOKA202488}; and  political governance \citep{dayanikli2025cooperation,chu2025mean}.
\end{itemize}


We argue that the fact that these studies have generally been modelling exercises rather than actually applied to real-world settings is at least in part due to classical algorithms being conceived in ways that may not be practical for real-world deployments (for example, they try to find solutions via analytical methods or oracles 
\citep{survey_learningMFGs})
. Though we do not claim to have finally solved this issue nor present real-world case studies ourselves, our work is motivated by the desire to remove some of the obstacles towards MFG algorithms that might be more realistic and practical.

For large, complex many-agent systems of physical decision makers deployed in the real world, such as swarm robotics or autonomous vehicle traffic, it may be infeasible to find MFG-NEs via analytical methods or oracles as they have been traditionally, such that learning must instead be conducted directly by an actual finite population in its deployed environment. In contrast to the restrictive assumptions of many previous methods, we argue that in such deployed scenarios desirable qualities for MFG algorithms include (these motivate the setting of the algorithms we present):

\begin{itemize}
    \item the ability to learn from an empirical distribution of $N$ agents (i.e. this distribution is generated only by the policies of the agents, rather than being updated or manipulated by the algorithm itself or an external oracle/simulator); 
    \item learning online from a single, non-episodic system run (also referred to in other works as a single sample path/trajectory \citep{zaman2023oraclefree,policy_mirror_independent}) - i.e. similar to the above, the population is not arbitrarily reset by an external controller, since it might be impractical to repeatedly reset the large deployed population; 
    \item  learning without reliance on a model of the system;
    \item  decentralisation;
    \item  fast practical convergence \citep{HUANG2025114057};
    \item  and robustness to unexpected failures of decentralised learners or changes in population size \citep{10100954}.
\end{itemize}

We give a detailed comparison with prior methods, which omit one or more of these desiderata, in Sec. \ref{related_work}. However we elaborate here on one quality in particular, namely the fact that almost all prior work relies on a centralised node to learn on behalf of all the agents. In this context `centralised' does not necessarily imply global observability of the whole population's actions - which would generally make computation infeasible given the complexity of the problem - but rather that learning is only conducted from the samples of a single representative agent, whose policy updates are assumed to be automatically pushed to the rest of the population by the central node \citep{Fornasier_Solombrino_2014,carmona2019linear,learningMFGs,doi:10.1073/pnas.1922204117,one_that_may_sample,survey_learningMFGs,angiuli2021unified,angiuli2023convergence,Anahtarci2020QLearningIR,zaman2023oraclefree,INOUE2023100217,policy_mirror_independent,lee2024mean,jeloka2025learning,yang2025discounted}. For this reason, whilst `centralised learning' is the term used in prior works, we generally refer to `central-agent learning' to reduce confusion. 

The use of a central learner in MFG algorithms naturally reflects the simplifying assumption of the framework, namely that since we are considering a limit distribution of symmetrical agents, we can consider a representative agent that interacts with this anonymous population. In central-agent algorithms, often the empirical mean field of the actual population is not even used to compute rewards or transitions, with the central learner instead updating an estimate of the mean field based only on its own policy, which is in turn used as input to its reward and transition functions \citep{carmona2019linear,angiuli2021unified,angiuli2023convergence}. 

However, recent works on MFGs, as in other areas of the multi-agent systems community, have recognised that the existence of a central coordinator is a very strong assumption in complex, real-world settings even without global observability, and one that can both restrict scalability by constituting a bottleneck for computation and communication, and reveal a vulnerable single point of failure for the whole system \citep{double,mainone,decentralised_review,marlreview,distributed_review,jiang2024fully,xu2025targets,AGYEMAN2025106183,10971233, policy_mirror_independent}. For example, if the single server coordinating all of a smart city's autonomous vehicles were to crash, the entire road network would cease to operate. 

As an alternative, some recent works have explored MFG algorithms involving the $N$ individual agents in the empirical population learning policies for themselves without relying on a central node \citep{parise2015network,7331083,7402908,7368131,Mguni_Jennings_Munoz_de_Cote_2018,9992399,subjective_equilibria,policy_mirror_independent,li2024incomplete,he2025hybrid}. However, those works do not meet one or more of our other desiderata for deployed algorithms: they generally focus on existence proofs for equilibria or theoretical sample guarantees, instead of practical convergence speed, and have largely not considered robustness in the senses we address, despite fault-tolerance being an original motivation behind many-agent systems - we compare with these other works more fully in the related work in Sec. \ref{related_work}.

\begin{figure*}[t]
    \centering
    \begin{subfigure}{0.329\textwidth}
        \centering
        \includegraphics[width=\linewidth]{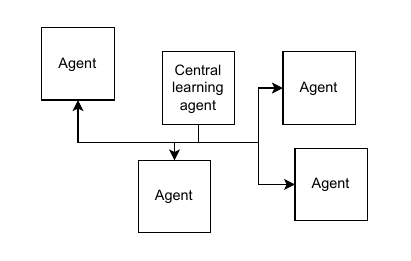}
        \caption{Central-agent.}
    \end{subfigure}
    \hfill
    \begin{subfigure}{0.329\textwidth}
        \centering
        \includegraphics[width=\linewidth]{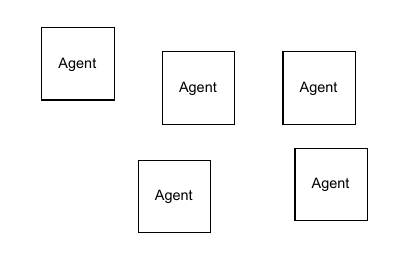}
        \caption{Independent.}
    \end{subfigure}
    \hfill
    \begin{subfigure}{0.329\textwidth}
        \centering
        \includegraphics[width=\linewidth]{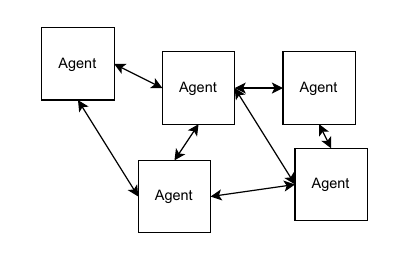}
        \caption{Decentralised network.}
    \end{subfigure}
    \caption{The three learning architectures for MFGs. The classical approach is the central-agent architecture, but this may be unrealistic in practice and presents a bottleneck and single point of failure. An independent learning architecture avoids these downsides, but has much worse theoretical sample guarantees and empirical learning speed due to bias caused by the divergence between independent agents' policies. We propose the decentralised networked architecture, which brings benefits over the other architectures in terms of both learning speed and robustness.}
    \label{Architectures_diagram}
\end{figure*}

Our main contribution is that we introduce decentralised, networked communication to the mean-field framework, allowing us to address all of our desiderata. Communication networks have had success in other multi-agent settings, removing the reliance on inflexible, centralised structures \citep{distributed_consensus,double,finite_consensus,comm_efficient_networked,finite_no_policy,off_policy,decentralised_review,marlreview,distributed_review}. We use the network in a scheme whereby agents can adopt policies communicated to them from neighbours. We show that as well as allowing us to learn without the assumption of access to centrally provided information, the communication network brings two important benefits. Firstly, populations using our networked communication scheme can learn faster than agents learning entirely independently. Secondly, the networked architecture affords robustness to unexpected failures of decentralised learners and changes in population size.
It may appear counter-intuitive to offer a policy communication algorithm in a setting that is non-cooperative. We offer a number of responses to pre-empt such concerns:

\begin{itemize}
    \item As already mentioned, most methods for solving MFGs involve a central learner pushing its policy to the rest of the population. This is also a type of communication, and there is no reason selfish agents should want to accept identical policies from a central node any more than they should want to selectively adopt policies communicated by neighbours.
    \item Prior work by \citet{policy_mirror_independent} compared the theoretical sample guarantees of central-agent and independent-learning algorithms for solving stationary MFGs which satisfy more of our desired characteristics than any other works, in that they learn online from the empirical distribution of the finite population from a non-episodic run of the empirical system. This presents an excellent opportunity for us to compare the sample guarantees of our networked algorithm, of which these two alternative architectures can be seen as special cases, within this existing theoretical context. 
    \item In our experiments we focus on \textit{coordination games}, where agents can increase their individual rewards by following the same strategy as others and therefore inherently have an incentive to communicate policies, even if the MFG setting itself is technically non-cooperative. Thus our work can be applied to real-world problems in e.g. traffic signal control, formation control in swarm robotics, and consensus and synchronisation e.g. for sensor networks \citep{10757965}. Nevertheless we find no need to make a distinction in our theoretical analysis, which holds across all types of non-cooperative MFG.
\end{itemize}

{We also pre-empt objections that communication with neighbours might violate the anonymity that is characteristic of the mean-field paradigm, by emphasising that the communication in our algorithm takes place outside of the ongoing learning-and-updating parts of each iteration. Thus the core learning assumptions of the mean-field paradigm are unaffected, as they essentially apply at a different level of abstraction (a convenient approximation) to the reality we face of $N$ agents that interact within the same environment. Indeed, prior works have combined networks with mean-field theory in different ways, such as using a mean field to describe adaptive dynamical networks \citep{berner2023adaptive}.
} 

In the idealised theoretical setting, we prove that our networked algorithm's theoretical sample guarantees lie between those of existing baseline centralised and independent algorithms. However we show empirically that these theoretical algorithms, although affording comparison with baseline sample guarantees, were not actually able to learn in practical time, and so we extend all three algorithms with experience replay buffers in order to compare the architectures experimentally. In our setting of learning from a continued, non-episodic run of the system, in which \citet{policy_mirror_independent} and \citet{subjective_equilibria} are the mostly closely related works, our experience replay buffer is a novel contribution. Much of the theoretical analysis in \citet{policy_mirror_independent} centres on methods to ensure the independence of samples that are collected along this continued system trajectory and used once before being discarded. This makes the inclusion of a buffer that is cycled through repeatedly not obvious a priori.

We show empirically that when the agents' Q-functions can be only roughly estimated due to fewer samples/updates, possibly leading to high variance in policy updates, then using the communication network to propagate better-performing policies through the population leads to faster learning than that achieved by agents learning entirely independently, which still hardly appear to learn at all. This is crucial in large complex environments that may be encountered in real applications, where the idealised hyperparameter choices (such as learning rates and numbers of iterations) required for the theoretical convergence guarantees will be infeasible in practice. As well as demonstrating our scheme's empirical benefits for learning speed, we conduct additional studies showing the advantages of communication for system robustness. In summary, our contributions include the following:

\begin{itemize}
    \item We prove that a theoretical version of our networked algorithm (Alg. \ref{networked_algorithm}) has sample guarantees bounded between those of central-agent and independent algorithms for learning with a non-episodic run of the empirical system. We provide the order of the difference in these bounds in terms of network structure and number of communication rounds, and contribute a policy-update stability guarantee (Sec. \ref{properties}).
    
    \item We show experimentally that all three theoretical algorithms do not seem to learn at all in any practical time (Sec. \ref{ablation_buffer_experiment_section}). We therefore modify all three (Alg. \ref{networked_algorithm_experience_replay}, Sec. \ref{practical_enhancements}) to make learning feasible by including an experience replay buffer, allowing us to {contribute the first empirical demonstrations of learning in all three architectures}. 

    \item Our experiments demonstrate that in practical settings our communication scheme can markedly benefit learning speed over the independent case, sometimes performing similarly to the centralised case while removing the restrictive assumption of the latter. We also show that via our practical modifications we can learn without enforcing several of the algorithms' other theoretical assumptions (a goal shared by other works on practical MFG algorithms \citep{cui2023learning}) (Sec. \ref{discussion}).

    \item  We provide ablations and additional empirical studies showing that our decentralised communication architecture brings {further benefits over both the central-agent and independent alternatives in terms of robustness to unexpected update failures and changes in population size}. For further discussion of the relevance of these scenarios in large multi-agent systems, see Sec. \ref{robustness_experiments}.
\end{itemize}

\begin{figure*}[t]
    \centering
    \includegraphics[width=\textwidth]{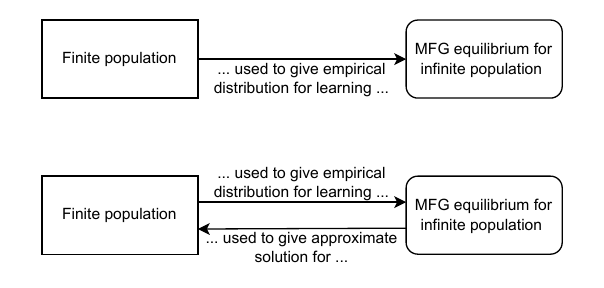}
    \caption{Two possible ways to conceive of our work regarding the relationship between the  infinite- and finite-population games. 
    Note that using the finite empirical population to try to learn a single MFG-NE policy $\boldsymbol\pi=(\pi^*,\dots,\pi^*)$ that is to be followed by the whole infinite population (Def. \ref{Nash_equilibrium_of_stationary_MFG}) is \textit{not} the same as directly finding $\boldsymbol{\pi}^* = (\pi^1,\dots,\pi^N)$, i.e. the tuple of \textit{individual} policies that gives the finite-population NE in Def. \ref{delta_NE}, a problem known to be hard \citep{nash_complexity,Vinyals2019GrandmasterLI,NEURIPS2020_e9bcd1b0,SHAVANDI2022118124,yardim2024exploiting,lirevisiting}.}
    \label{conceptions_diagram}
\end{figure*}

The paper structure is as follows: we give further related work in Sec. \ref{related_work} and preliminaries in Sec. \ref{preliminaries}. We present our theoretical algorithms in Sec. \ref{algorithm_section} and theoretical results in Sec. \ref{properties}. We give enhancements to the algorithms necessary for learning in practical time in Sec. \ref{practical_enhancements}, and provide experiments and discussion in Sec. \ref{experiments}. We conclude by discussing limitations and suggestions for future work in Sec. \ref{future_work}.

\begin{remark}\label{two_conceptions}
Solving the theoretical MFG involves finding the single policy that, when given to all agents in the infinite population, best responds to the resulting mean-field distribution. We 
give two ways to conceive of our work (illustrated in Fig. \ref{conceptions_diagram}), which mirror and make more explicit the similar motivations underpinning many other MFG works \citep{cui2023learningdecentralisedcontrol,dayanikli2024deep,zaman2024robust,bayraktar2024learning,subjective_equilibria,10928821,jeloka2025learning,cecchin2025weak,bo2025mean,tchuendom2025ranking,magninolearning,graber2025trembling,aggarwal2025distributed,yardim2025variational,hofer2025markov,feng2025decentralized,si2025decentralized,10928821}. Firstly, we contribute algorithms that allow the solution to a MFG to be learnt using the empirical distribution of a decentralised finite population, without needing to make unrealistic assumptions about access to an oracle for the infinite population. Note that it is unnecessary and possibly impractical to assume that the decentralised agents always follow a single identical policy throughout training, a logic also followed by earlier works \citep{policy_mirror_independent}. 

Alternatively, we may have originally been interested in finding a NE for a large, finite population, but, due to the scalability issues of learning approaches like MARL, forced to turn to the MFG framework to find a policy that gives an approximate solution to the finite-population problem. We contribute algorithms that allow the deployed finite population to find the MFG solution that in turn approximately solves the original problem, without unrealistic assumptions about centralised training. Under this framing, it may matter less whether all agents follow a single policy throughout training.
\end{remark}

\section{Related work}\label{related_work}

In Sec. \ref{introduction} we gave several qualities that we argue are desirable for mean-field algorithms if they are eventually to be more applicable to complex, deployed scenarios in the real-world. Conversely, works on MFGs have traditionally been largely theoretical \citep{huangMFG,lasry2007mean} (often works do not present any empirical results \citep{policy_mirror_independent,li2025effect,huang2025well,ferreira2025solving,lascu2025non}), and methods for finding equilibria have often relied on assumptions that are too strong for real-world applications. The MFG-NE is classically found by solving a coupled system of dynamical equations: a forward evolution equation for the mean-field distribution, and a backwards equation for the representative agent's optimal response to the mean field, as in Def. \ref{Nash_equilibrium_of_stationary_MFG} below \citep{YOSHIOKA202488,wang2024variational,li2024incomplete,10637582,CHEN2024409,10571600,si2024backward,federico2024mean,lee2024mean,10928821,aerospace12040302,cecchin2025weak,Sun14042025,dayanikli2025cooperation,wang2025linear,yang2025gaussian,ersland2025long,ghosh2025mean,aydin2025fare,yang2025discounted,pande2025generative,cao2025probabilistic,osborne2025rates,Martinez_Garcia_2025,opper2025mean,carlini2025semi,plank2025policy,hua2025extended,moll2025mean,ferreira2025solving,Progressive_Maximum_2025,tchuendom2025ranking,DEY2025122432,fedorov_boundary_2025,10907167,WANG2025112260,LI2026129830,xiang2025robust,li_diffusion_2025,ghosh2025federated,xu2025linear,si2025decentralized}; crucially, these methods generally rely on the assumption of an infinite population \citep{survey_learningMFGs}. Early work solved the coupled equations using numerical methods that did not scale well for more complex state and action spaces \citep{doi:10.1137/090758477,doi:10.1137/120902987,Briceño_Arias,achdou2020mean}; or, even if they could handle higher-dimensional problems, the methods were based on known models of the environment's dynamics (i.e. they were model-based) \citep{NEURIPS2019_030e65da, 10.3389/fams.2020.00011, cao2021connecting, carmona2021deep,germain2022numerical,Anahtarci2020QLearningIR,huang2024modelbased,pmlr_v238_huang24a,barreiro2025optimal}, and/or computed a best-response to the mean-field distribution \citep{huangMFG, NEURIPS2019_030e65da,Elie_Prolat_Laurire_Geist_Pietquin_2020,continuous_fictitious_play,flock,survey_learningMFGs,scalable_deep,algumaei2023regularization}. The latter approach is both computationally inefficient in non-trivial settings \citep{survey_learningMFGs,policy_mirror_independent}, and in many cases is not convergent (as in general it does not induce a contractive operator) \citep{approximately_entropy,scalable_deep}. Subsequent work, including our own, has therefore moved towards model-free and/or policy-improvement scenarios \citep{rl_stationary,9304340,Cacace,perolat2021scaling,pmlr_v130_lee21b,survey_learningMFGs,angiuli2021unified,10093067,general_framework}, possibly with learning taking place by observing \textit{N}-agent \textit{empirical} population distributions \citep{subjective_equilibria,policy_mirror_independent,hu2024mfoml}. 

Most prior works, including algorithms designed to solve MFGs using an \textit{N}-agent empirical distribution, have also assumed an oracle that can generate samples of the game dynamics (for any distribution) to be provided to the learning agent \citep{Anahtarci2019FittedQI,fu2019actor,NEURIPS2019_030e65da,general_framework,Anahtarci2020QLearningIR}, or otherwise that the algorithm (rather than agents' policies) has direct control over the population distribution at each time step \citep{10.5555/3545946.3598748,chen2024sampled,zhang2024stochastic}, such as cases where the agents' policies and distribution are updated on different timescales \citep{angiuli2023convergence,zeng2024single}, with the \textit{fictitious play} method being particularly popular \citep{tembine2012meanfield,Cardaliaguet2017,Mguni_Jennings_Munoz_de_Cote_2018,rl_stationary,continuous_fictitious_play,flock,one_that_may_sample,geist2022concave,bonnans2021generalized,lauriere2021numerical,angiuli2021unified,nips2022mao,scalable_deep,zaman2023oraclefree,cui2023learning,yu2024convergence}. In practice, many-agent problems may not admit such arbitrary generation or manipulation (for example, in the context of robotics or controlling vehicle traffic), and so a desirable quality of learning algorithms is that they update only the agents' policies, rather than being able to arbitrarily reset their states. Learning may thus also need to leverage continuing, rather than episodic, tasks \citep{sutton2018reinforcement}. \citet{policy_mirror_independent}, \citet{subjective_equilibria} and our own work therefore present algorithms that seek the MFG-NE using only a single run of the empirical population. 


Naturally, decentralised communication is most applicable in settings where learning takes place along a continuing system run, rather than the distribution being manipulated by an oracle or arbitrarily reset for new episodes, since these imply a level of external control over the population that results in centralised learning. Equally, it is in situations of learning from finite numbers of real, deployed agents (rather than settings able to simulate infinite populations) that we are most likely to be concerned with fault tolerance. Networked communication therefore naturally fits within our desired qualities for mean-field algorithms, and our focus on this setting means that our work is most closely related to \citet{policy_mirror_independent} and \citet{subjective_equilibria}, which provide algorithms for centralised and independent learning with empirical distributions along non-episodic system runs. We contribute a networked learning algorithm in this setting. 
\citet{subjective_equilibria} empirically demonstrates an independent learning algorithm when agents observe compressed information about the mean-field distribution as well as their local state, but they do not compare this to any other algorithms or baselines. \citet{policy_mirror_independent} compares algorithms for centralised and independent learning theoretically, but does not provide empirical demonstrations. In contrast, in addition to providing theoretical guarantees, we empirically demonstrate our networked learning algorithm, where agents observe only their local state, in comparison to both centralised and independent baselines, as well as concerning ourselves with the speed of practical convergence and robustness, unlike these works. 

More generally, a number of works refer to `decentralisation' in MFGs, but often in a different sense to our understanding of it. In particular, many works that say they consider decentralisation actually learn/derive policies via a centralised method (often involving a representative player), and simply mean that agents' policies are \textit{executed} independently based on local information, which we take as a given across our learning architectures \citep{WANG2025112260,choutri2025backpressure,xiang2025robust,feng2025decentralized,si2025decentralized,si2025general}. \citet{he2025hybrid} use RL to solve a two-level mean-field problem, where there is a MFG between `aggregators', each of which is solving a local mean-field control (MFC) problem (the cooperative alternative to a MFG). They solve the MFG via decentralised learning by the $N$ aggregators, but each aggregator solves its MFC problem in a centralised manner via the assumption of a single agent that is representative of the heterogenous population. Moreover, they prove the existence of and convergence to a unique equilibrium, but do not provide sample guarantees or a convergence rate, as we do in Sec. \ref{properties}. Other works involve decentralisation in learning but under different MFG settings to our own: \citet{li_diffusion_2025,ghosh2025federated,xu2025linear} derive controls in a decentralised way, but rely on a model of the environment, while \citet{yardim2025variational} uses independent learning but not via RL, as they focus on repeated play of static, stateless games. 


Improving the training speed and sample efficiency of (deep) (multi-agent) RL is gaining increasing attention \citep{wiggins2023characterizing,yu2024cheaper, 10502122,patel2024global}, though our own work is one of the only on MFGs to be concerned with this. \citet{huang2024unsupervised} trains on a distribution of MFG configurations to speed up inference on unseen problems, but does not learn online in a decentralised manner as in our own work. 
 Similarly, while some attention has been given to the robustness of multi-agent systems to changes in population size, where it is sometimes referred to as `ad-hoc teaming', `open-agent systems', `scalability' or `generalisation' \citep{10.1002/aaai.12131}, it has more commonly been addressed in MARL \citep{dawood2023safe,10.1007/978-981-97-1087-4_9} than in MFGs \citep{wu2024populationaware}. \citet{wu2024populationaware} presents an MFG approach that allows new agents to join the population during \textit{execution}, but training itself takes place offline in a centralised, episodic manner. Our networked communication framework, on the other hand, allows decentralised agents to join the empirical population during online learning and to have minimal impact on the learning process by adopting policies from existing members of the population through communication (Sec. \ref{robustness_experiments}).  


An existing area of work called \textit{robust mean-field games} studies the robustness of these games to uncertainty in the transition and reward functions \citep{BAUSO2012454,bauso2016robust,7300401,7488259,doi:10.1137/15M1014437,8270667,9867331,aydin2023robustness}, but does not consider resilience to agent update failures, despite fault tolerance being one of the original motivations behind many-agent systems. On the other hand, we focus on robustness to failures and changes in the agent population itself.

For the sake of defining the scope of terms, we do not consider what we refer to as the `mean-field framework' to encompass the related but distinct research area called \textit{mean-field RL} \citep{mfrl_yang,pomfrl,subramanian2022multi}. While drawing inspiration from similar sources, mean-field RL falls outside of the game-theoretic MFG paradigm. It is instead a type of MARL, and considers a mean over actions (originally by averaging pairwise interactions between agents \citep{mfrl_yang}) rather than a distribution over states, as in our case. This generally leads to lag and a chicken-and-egg problem, whereby agents respond to the other agents' previous actions, rather than their current states. The existence of this distinct area can be a source of confusion in nomenclature: while we use RL as a model-free learning approach in MFGs, this is not the same as doing mean-field RL, and as we discuss above MFGs can be, and classically were, solved without RL. Some works have considered similar features to those we are interested in, such as decentralisation and estimation from local neighbourhoods, but in this distinct area of mean-field RL \citep{pomfrl}. We draw attention in particular to the work by \citet{DMFG}, who develop a framework they refer to as `Decentralized Mean Field Games [sic]' but which they also emphasise is distinct from both MFGs and mean-field RL, and which removes the symmetry and anonymity of agents (agents may have different reward functions, and agent indices are retained). Therefore, despite its name, this setting is different from our own.

We note a similarity between 1. our method for deciding which policies to propagate through the population (described in Sec. \ref{practical_generation}) and 2. the computation of evaluation/fitness functions within evolutionary algorithms to indicate which solutions are desirable to keep in the population for the next generation \citep{Eiben2015,sissodia2025evolutionary}. Moreover, the research avenue broadly referred to as \textit{distributed embodied evolution} involves swarms of agents independently running evolutionary algorithms while operating within a physical/simulated environment and communicating behaviour parameters to neighbours \citep{haasdijk2014combining,10.1145/2739482.2768490}, and is therefore even more similar to our setting, where decentralised RL updates are computed locally and then shared with neighbours. In distributed embodied evolution, the computed fitness of solutions helps determine both which are preserved by agents during local updates, and also which are chosen for broadcast or adoption between neighbours \citep{survivability,Maintaining_Diversity,10.1007/978-3-030-16692-2_38}. Indeed, some works on distributed embodied evolution specifically consider features or rewards relating to the joint behaviour of the whole population \citep{10.1145/2463372.2463398,prieto2016real}, similar to MFGs. The adjacent research area of cultural/language evolution for swarm robotics \citep{Cultural_Propagation,10.3389/frobt.2020.00012,CAMBIER2021108010} has similarly demonstrated the combination of evolutionary approaches and multi-agent communication networks for self-organised behaviours in swarms. However, unlike our own work, none of these areas employ RL in the update of policies or the computation of the fitness functions.

Our work also shares parallels with \textit{population-based training} \citep{jaderberg2017population}, an approach that is likewise related to evolutionary algorithms. Population-based training involves optimising neural networks by performance-based transfer of parameters and hyperparameters among a population of concurrent processes. Our algorithms are tabular\footnote{This is because we give theoretical comparisons of our networked architecture with the previous SOTA algorithms for the central-agent and independent architectures \citep{policy_mirror_independent}, which were also tabular. Nevertheless this potentially limits scalability, and we discuss the extension to non-tabular algorithms in the future work in Sec. \ref{future_work}.} rather than neural network-based, and we are also interested in the interactive behaviour of the population itself rather than simply using it for parallelising the optimisation.



\section{Preliminaries}\label{preliminaries}

We use the following notation. $N$ is the number of agents in a population, with $\mathcal{S}$ and $\mathcal{A}$ representing the finite state and common action spaces, respectively. {The sets $\mathcal{S}$ and $\mathcal{A}$ are equipped with the discrete metric $d(x,y) = {\mathds{1}}_{x \neq y}$}. The set of probability measures on a finite set $\mathcal{X}$ is denoted $\Delta_\mathcal{X}$, and $\mathbf{e}_x \in \Delta_\mathcal{X}$ for $x \in \mathcal{X}$ is a one-hot vector with only the entry corresponding to $x$ set to 1, and all others set to 0. For time $t \geq 0$, $\hat{\mu}_t$ = $\frac{1}{N}\sum^N_{i=1}\sum_{s\in\mathcal{S}}$ $\mathds{1}_{s^i_t=s}\mathbf{e}_s$ $\in$ $\Delta_\mathcal{S}$ is a vector of length $|\mathcal{S}|$ denoting the empirical categorical state distribution of the $N$ agents at time $t$. The set of policies is $\Pi$ = \{$\pi$ : $\mathcal{S} \rightarrow \Delta_\mathcal{A}$\}, and the set of Q-functions is denoted $\mathcal{Q} = \{q : \mathcal{S} \times \mathcal{A} \rightarrow \mathbb{R}\}$.   
For $\pi,\pi' \in \Pi$ and $q,q' \in \mathcal{Q}$, we have the norms $||\pi - \pi'||_1$ := $\sup_{s \in \mathcal{S}}||\pi(s) - \pi'(s)||_1$ and $||q - q'||_{\infty}$ := $\sup_{s\in\mathcal{S},a\in\mathcal{A}}|q(s,a) - q'(s,a)|$. 

Function $h : \Delta_{\mathcal{A}} \rightarrow \mathbb{R}_{\geq0}$ denotes a strongly concave function, which we implement in our experiments as the scaled entropy regulariser $\lambda h_{ent}(u) = - \lambda\sum_a u(a)\log u(a)$, for $a \in \mathcal{A}$, $u \in \Delta_\mathcal{A}$ and $\lambda > 0$. As in many earlier works \citep{approximately_entropy,10.1287/moor.2021.1238,Anahtarci2020QLearningIR,algumaei2023regularization,yu2023time,policy_mirror_independent,yardim2025variational,lu2025convergence,ferreira2025solving,lascu2025non,he2025hybrid}, regularisation is theoretically required to ensure the contractivity of operators and continued exploration, and hence algorithmic convergence. However, it has been recognised that modifying the RL objective in this way can bias the NE \citep{scalable_deep,pmlr-v162-su22b,policy_mirror_independent, hu2024mfoml,lu2025convergence}. We show in our experiments that we are able to reduce $\lambda$ to 0 with no detriment to convergence.

We now define, for $h_{\mathrm{max}} > 0$ and $h : \Delta_{\mathcal{A}} \rightarrow [0,h_{\mathrm{max}}]$, $u_{\mathrm{max}} \in \Delta_{\mathcal{A}}$ such that $h(u_{\mathrm{max}}) = h_{\mathrm{max}}$. We further define $Q_{\mathrm{max}} := \frac{1 + h_{\mathrm{max}}}{1 - \gamma}$, and set $\pi_{\mathrm{max}} \in \Pi$ such that $\pi_{\mathrm{max}}(s) = u_{\mathrm{max}}, \forall s \in \mathcal{S}$. For any $\Delta h \in \mathbb{R}_{>0}$, we also define the convex set $\mathcal{U}_{\Delta h}$ := \{$u \in \Delta_{\mathcal{A}} : h(u) \geq h_{\mathrm{max}} - \Delta h\}$. 

Having given this notation, we now formalise the symmetric anonymous game involving $N$ agents.

\begin{definition}[$N$-player symmetric anonymous games] 
An N-player stochastic game with symmetric, anonymous agents is given by the tuple $\langle$$N$, $\mathcal{S}$, $\mathcal{A}$, $P$, $R$, $\gamma$$\rangle$, where $\mathcal{A}$ is the action space, identical  for each agent; $\mathcal{S}$ is the identical  state space of each agent, such that their  initial states are \{$s^i_0$\}$_{i=1}^N \in \mathcal{S}^N$ and their policies are \{$\pi^i$\}$_{i=1}^N \in \Pi^N$. $P$ : $\mathcal{S}$ $\times$ $\mathcal{A}$ $\times$ $\Delta_{\mathcal{S}}$ $\rightarrow$ $\Delta_{\mathcal{S}}$ is the transition function and $R$ : $\mathcal{S}$ $\times$ $\mathcal{A}$ $\times$ $\Delta_{\mathcal{S}}$ $\rightarrow$ [0,1] is the reward function,  which map each agent's local state and action and the population's empirical distribution to transition probabilities and bounded rewards, respectively, i.e. $\forall i \in \{1,\dots,N\}$\[
s^i_{t+1} \sim P(\cdot|s^i_{t},a^i_{t},\hat{\mu}_t) \;\;\;\; and \;\;\;\; r^i_{t} = R(s^i_{t},a^i_{t},\hat{\mu}_t).\]\end{definition}
The policy of an agent is given by $a^i_t\sim \pi^i(s^i_t)$, that is, each agent only observes its own state, and not the joint state or empirical distribution of the population.  

We now formalise the expected discounted returns of each agent in an $N$-player symmetric anonymous game.

\begin{definition}[$N$-player discounted regularised return]\label{psi_definition}
With joint policies $\boldsymbol\pi$ := ($\pi^1,\dots,\pi^N$) $\in \Pi^N$, initial states sampled from a distribution $\upsilon_0 \in \Delta_\mathcal{S}$ and $\gamma$ $\in$ [0,1) as a discount factor, the expected discounted regularised returns of each agent $i$ in the symmetric anonymous game are given by, $\forall i,j \in \{1,\dots,N\}$,  
\begin{align*} 
\Psi^i_h(\boldsymbol\pi,\upsilon_0) =  \mathbb{E}\left[\sum^{\infty}_{t=0}\gamma^{t}(R(s^i_t,a^i_t,\hat{\mu}_t) + 
h(\pi^i(s^i_t)))\bigg|\substack{s^j_0\sim\upsilon_0 \\ a^j_t\sim \pi^j(s^j_t)\\ s^j_{t+1} \sim P(\cdot|s^j_{t},a^j_{t},\hat{\mu}_t)}\right].  
\end{align*}
\end{definition}

This allows us to formalise the solution concept for the competitive $N$-player symmetric anonymous game, namely the (approximate) NE.

\begin{definition}[$\delta$-NE]\label{delta_NE}
Say $\delta > 0$ and $(\pi,\boldsymbol\pi^{-i})$ := $(\pi^1,\dots\pi^{i-1},\pi,\pi^{i+1},\dots,\pi^{N}) \in \Pi^N$. An initial distribution $\upsilon_0 \in \Delta_\mathcal{S}$ and an $N$-tuple of policies $\boldsymbol\pi$ := ($\pi^1,\dots,\pi^N$) $\in \Pi^N$ form a $\delta$-NE ($\boldsymbol\pi$, $\upsilon_0$) if 
\[\Psi^i_h(\boldsymbol\pi,\upsilon_0) \geq \max_{\pi \in \Pi}\Psi^i_h((\pi,\boldsymbol\pi^{-i}),\upsilon_0) - \delta \;\;\; \forall i \in \{1,\dots, N\}.  \] 
\end{definition}
At the limit as $N \rightarrow \infty$, the population of infinitely many agents can be characterised as a limit distribution $\mu \in \Delta_\mathcal{S}$. We denote the expected discounted return of the representative agent in the infinite-agent game - termed a MFG - as $V$, rather than $\Psi$ as in the finite $N$-agent case. 
\begin{definition}[Mean-field discounted regularised return]\label{mean_field_return_def}
For a policy-population pair ($\pi, \mu$) $\in \Pi \times \Delta_\mathcal{S}$, 
\begin{align*}
V_h(\pi,\mu) = \mathbb{E}\left[
\sum^{\infty}_{t=0}\gamma^{t}(R(s_t,a_t,\mu)+h(\pi(s_t)))\bigg|
\substack{s_0\sim\mu\\ 
a_t\sim \pi(s_t)\\ s_{t+1} \sim P(\cdot|s_{t},a_{t},{\mu})}\right].
\end{align*}
\end{definition}
A stationary MFG is one that has a unique population distribution that is stable with respect to a given policy, and the agents' policies are not time- or population-dependent. We now introduce the solution concept of this stationary MFG.

\begin{definition}[NE of stationary MFG]\label{Nash_equilibrium_of_stationary_MFG}
For a policy $\pi^* \in \Pi$ and a population distribution $\mu^* \in \Delta_\mathcal{S}$, the pair ($\pi^*,\mu^*$) is a stationary MFG-NE if the following optimality and stability conditions hold: \begin{align*}
    & \textrm{optimality:} \quad V_h(\pi^*,\mu^*) = \max_\pi V_h(\pi,\mu^*), \\
     & \textrm{stability:} \quad \mu^* (s) = \sum_{s',a'}\mu^* (s')\pi^* (a'|s')P(s|s',a',\mu^*).
\end{align*} If the optimality condition is only satisfied with $V_h(\pi^*_\delta,\mu^*_\delta) \geq \max_\pi V_h(\pi,\mu^*_\delta) - \delta$, then ($\pi^*_\delta,\mu^*_\delta$) is a $\delta$-NE of the MFG, where $\mu^*_\delta$ is obtained from the stability equation and $\pi^*_\delta$.  
\end{definition}


The MFG-NE is an approximate NE of the finite $N$-player game, in which we may have originally been interested but which is difficult to solve in itself \citep{scalable_deep,policy_mirror_independent}: 

\begin{proposition}[\textit{N}-player NE and MFG-NE (Thm. 1, \citep{Anahtarci2020QLearningIR})] 
If ($\pi^*,\mu^*$) is a MFG-NE, then, under certain Lipschitz conditions \citep{Anahtarci2020QLearningIR}, for any $\delta > 0$, there exists $N(\delta) \in \mathbb{N}_{>0}$ such that, for all $N \geq N(\delta)$, the joint policy $\boldsymbol\pi = \{\pi^*,\pi^*,\dots,\pi^*\} \in \Pi^N$ is a $\delta$-NE of the $N$-player 
game. 
\end{proposition}

\begin{remark}\label{1overrootN}
We can show that $\delta$ can be characterised further in terms of $N$, with ($\pi^*,\mu^*$) being an $\mathcal{O}$($\frac{1}{\sqrt{N}}$)-NE of the $N$-player symmetric anonymous game \citep{policy_mirror_independent,chen2024periodic,yardim2025variational}. 
\end{remark}

For our new, networked learning algorithm, we also introduce the concept of a time-varying communication network, where the links between agents that make up the network may change at each time step $t$. Most commonly we might think of such a network as depending on the spatial locations of decentralised agents, such as physical robots, which can communicate with neighbours that fall within a given broadcast radius. When the agents move in the environment, their neighbours and therefore communication links may change. However, the dynamic network is general to all settings, and can depend on other factors that may not depend on the agents' position in space or state $s^i_t$. For example, agents may be connected over long distances via satellites or the internet, and even a network of fixed-location agents can change depending on which agents are active and broadcasting at a given time $t$, or if their broadcast radius changes, perhaps in relation to signal or battery strength. 

\begin{definition}[Time-varying communication network] 
The time-varying communication network ($\mathcal{G}_t$)$_{t\geq0}$ is given by $\mathcal{G}_t$ = ($\mathcal{N}, \mathcal{E}_t$), where $\mathcal{N}$ is the set of vertices each representing an agent $i \in \{1,\dots,N\}$, and the edge set $\mathcal{E}_{t}$ $\subseteq$ \{(\textit{i},\textit{j}) : \textit{i},\textit{j} $\in$ $\mathcal{N}$, \textit{i} $\neq$ \textit{j}\} is the set of undirected communication links by which information can be shared at time \textit{t}. A network's \textit{diameter} $d_{\mathcal{G}_t}$ is the maximum of the shortest path length between any pair of nodes.
\end{definition}

A network is \textit{connected} if there is a sequence of distinct edges forming a path between each distinct pair of vertices. The \textit{union} of a collection of graphs \{$\mathcal{G}_t, \mathcal{G}_{t+1}, \cdots, \mathcal{G}_{t+\omega}\}$ ($\omega \in \mathbb{N}$) is the graph with vertices and edge set equalling the union of the vertices and edge sets of the graphs in the collection \citep{1205192}. A collection is \textit{jointly connected} if its members' union is connected. 

\subsection{Further technical conditions for algorithms and theorems
}\label{Further_definitions_and_assumptions}

Our theoretical results, which compare our networked algorithm with the centralised and independent alternatives from \citet{policy_mirror_independent}, rely on several further definitions and assumptions from their work. We give these now to allow us to introduce our learning operator for our algorithm in Sec. \ref{algorithm_section}, in advance of the theoretical analysis in Sec. \ref{properties}. These formalisations lay the groundwork that allows the optimality and stability conditions, which define the MFG-NE in Def. \ref{Nash_equilibrium_of_stationary_MFG}, to hold.



\subsubsection{Population update operators}

The following characterisations ensure that the evolution of the population is convergent for a given policy. 

Assumption \ref{lip_cont_assumpt} gives Lipschitz constants that provide smoothness conditions on the transition and reward functions $P$ and $R$ - this is a standard assumption in previous work \citep{policy_mirror_independent}. These in turn ensure that the population-evolution and policy-update operators below are smooth and hence contractive, guaranteeing convergence.
\begin{assumption}[Lipschitz continuity of $P$ and $R$
]\label{lip_cont_assumpt}
There exist constants $K_{\mu},K_{s},K_{a},L_{\mu},L_{s},L_{a} \in \mathbb{R}_{\geq 0}$ such that $\forall s,s' \in \mathcal{S},\forall a,a' \in \mathcal{A},\forall \mu,\mu' \in \Delta_{\mathcal{S}}$,
\[
||P(\cdot|s,a,\mu) - P(\cdot|s',a',\mu')||_{1} \leq K_{\mu}||\mu - \mu'||_{1} + K_{s}d(s,s') + K_{a}d(a,a'),
\]
\[
|R(s,a,\mu) - R(s',a',\mu')| \leq L_{\mu}||\mu - \mu'||_{1} + L_{s}d(s,s') + L_{a}d(a,a'). 
\]
\end{assumption}

The following single-step operator tells us how the mean field evolves by one step when the whole population uses a certain policy, which allows us in turn to give the stable population operator as the fixed point of repeated updates. This is later plugged into the policy-improvement operator, allowing us to obtain the fixed-point consistency and hence the stationary MFG-NE.

\begin{definition}[Population update operator
]
The single-step population update operator $\Gamma_{pop} : \Delta_{\mathcal{S}} \times \Pi \rightarrow \Delta_{\mathcal{S}}$ is defined as, $\forall s \in \mathcal{S}$: \[\Gamma_{pop}(\mu,\pi)(s) := \sum_{s'\in\mathcal{S}}\sum_{a'\in\mathcal{A}}\mu(s')\pi(a'|s')P(s|s',a',\mu).\]
We will use the short hand notation $\Gamma^n_{pop}(\mu,\pi) := \underbrace{\Gamma_{pop}(\dots\Gamma_{pop}(\Gamma_{pop}(\mu,\pi),\pi),\dots,\pi)}_{n\text{  times}}$. 
\end{definition}

We recall in the following lemma that $\Gamma_{pop}$ is known to be Lipschitz \citep{learningMFGs,Anahtarci2020QLearningIR}. By ensuring that the population updates are smooth, we can in turn ensure that they are contractive, giving a unique and stable steady population via Assumption \ref{stable_population} below. This in turn ensures convergence. 

\begin{lemma}[Lipschitz population updates
]
$\Gamma_{pop}$ is Lipschitz with \[
||\Gamma_{pop}(\mu,\pi) - \Gamma_{pop}(\mu',\pi')||_{1} \leq L_{pop,\mu}||\mu - \mu'||_{1} + \frac{K_a}{2}||\pi -  \pi'||_1,
\]

where $L_{pop,\mu}$ := $\left(\frac{K_s}{2} + \frac{K_a}{2} + K_{\mu}\right)$, $\forall \pi \in \Pi, \mu \in \Delta_{\mathcal{S}}$. 
\end{lemma}

For stationary MFGs the population distribution must be stable with respect to a policy, requiring that $\Gamma_{pop}(\cdot,\pi)$ is contractive $\forall \pi \in \Pi$. We therefore give the following assumption, which is common in previous works \citep{learningMFGs, policy_mirror_independent,Anahtarci2020QLearningIR,zaman2023oraclefree}:

\begin{assumption}[Stable population
]\label{stable_population}
Population updates are stable, i.e. $L_{pop,\mu} < 1$.
\end{assumption}

Assumption \ref{stable_population} allows us to give the following operator that maps policies to their stable mean-field distributions.

\begin{definition}[Stable population operator $\Gamma_{pop}^{\infty}$
]
Given Assumption \ref{stable_population}, the operator $\Gamma_{pop}^{\infty} : \Pi \rightarrow \Delta_{\mathcal{S}}$ maps a given policy to its unique stable population distribution such that $\Gamma_{pop}(\Gamma_{pop}^{\infty}(\pi),\pi) = \Gamma_{pop}^{\infty}(\pi)$, i.e. the unique fixed point of $\Gamma_{pop}(\cdot,\pi) : \Delta_{\mathcal{S}} \rightarrow \Delta_{\mathcal{S}}.$
\end{definition}

\subsubsection{Policy improvement operators} 

We now introduce the policy improvement operators, which we use for policy improvement in place of a direct best-response operator (see Sec. \ref{related_work}). As with the population evolution operators, these must also be Lipschitz, to ensure smoothness and hence convergence.

We first define the regularised Q-functions.

\begin{definition}[$Q_h$ and $q_h$ functions] 
We define, for any pair $(s,a) \in \mathcal{S} \times \mathcal{A}$:
 \[Q_h(s,a|\pi,\mu) := \mathbb{E}\left[\sum^{\infty}_{t=0}\gamma^t(R(s_t,a_t,\mu) + 
 h(\pi(s_t)))\left|\substack{s_0 = s, \\ a_0 = a}, \substack{s_{t+1} \sim P(\cdot|s_t,a_t,\mu), \\ a_{t+1} \sim \pi(\cdot|s_{t+1})}, \forall t \geq 0 \vphantom{\sum^{\infty}_{t=0}\gamma^t(R(s_t,a_t,\mu) +} \right] \right. \]
and 
    \[q_h(s,a|\pi,\mu) := R(s,a,\mu) + \gamma\sum_{s',a'}P(s'|s,a,\mu)\pi(a'|s')Q_h(s',a'|\pi,\mu).\]
\end{definition}

We can now give the operator that maps policy-population pairs to Q-functions, i.e. it gives the Q-function when an agent uses a certain policy when the population has a certain mean-field distribution. We are able to approximate this operator via Def. \ref{CTD_operator} below, which learns online from samples taken along a trajectory of the current policy.

\begin{definition}[$\Gamma_q$ operator]\label{gamma_q_operator} 
The operator $\Gamma_q : \Pi \times \Delta_{\mathcal{S}} \rightarrow \mathcal{Q}$, which maps policy-population pairs to Q-functions, is defined as $\Gamma_q(\pi,\mu) := q_h(\cdot,\cdot|\pi,\mu) \in \mathcal{Q}$ $\forall \pi \in \Pi, \mu \in \Delta_{\mathcal{S}}$.
\end{definition}

We now define the policy mirror ascent (PMA) operator for policy improvement. Agents update a policy with respect to a given Q-function by selecting, for each state, a probability distribution over their actions that maximises the combination of three terms (Def. \ref{PMA_operator}): 1. the value of the given state with respect to the Q-function; 2. a regulariser over the action probability distribution (in practice, we maximise the scaled entropy of the distribution); 3. a metric of similarity between the new action probabilities for the given state and those of the previous policy, given by the squared two-norm of the difference between the two distributions. We can alter the importance of the similarity metric relative to the other two terms by varying a parameter $\eta$, which is equivalent to changing the learning rate of the policy update. The three terms in the maximisation function can be seen in the PMA operator:

\begin{definition}[Policy mirror ascent operator (Def. 3.5, \citep{policy_mirror_independent})]\label{PMA_operator}
For a learning rate $\eta > 0$ and $L_h := L_a + \gamma\frac{L_{s}K_{a}}{2-\gamma K_{s}}$ (where these constants are defined in Assumption \ref{lip_cont_assumpt}), the PMA update operator $\Gamma^{md}_\eta : \mathcal{Q} \times \Pi \rightarrow \Pi$ is defined as, $\forall s \in \mathcal{S}, \forall Q \in \mathcal{Q}, \forall \pi \in \Pi$
\begin{align*}
\Gamma^{md}_\eta(Q,\pi)(s) :=\underset{u \in \mathcal{U}_{L_h}}{\arg\max}\left(\langle u,q(s,\cdot)\rangle + h(u) - \frac{1}{2\eta}||u - \pi(s)||^2_2 \right).
\end{align*}\end{definition}

$\Gamma_q$ and $\Gamma^{md}_\eta$ are both known to be Lipschitz continuous \citep{policy_mirror_independent}.

We can now define the theoretical learning operator $\Gamma_{\eta}$, which is used in the fixed-point iterations to find the MFG-NE. $\Gamma_{\eta}$ takes a PMA step to update a policy with respect to the Q-function of that policy when the population has a stable mean-field distribution arising from following that policy.

\begin{definition}[Nested learning operator] 
For a learning rate $\eta > 0$, $\Gamma_{\eta} : \Pi \rightarrow \Pi$ is defined as
 \[\Gamma_{\eta}(\pi) := \Gamma_{\eta}^{md}(\Gamma_{q}(\pi,\Gamma_{pop}^{\infty}(\pi)),\pi).\]
\end{definition}

The following lemma demonstrates that the fixed points of $\Gamma_{\eta}$ are MFG-NE policies.

\begin{lemma}[Fixed points of $\Gamma_{\eta}$ are MFG-NE] For arbitrary $\eta>0$, a pair ($\pi^*,\mu^*$) is a MFG-NE if and only if $\pi^*= \Gamma_{\eta}(\pi^*)$ and $\mu^* = \Gamma_{pop}^{\infty}(\pi^*)$.
\end{lemma}

We now recall that $\Gamma_{\eta}$ is Lipschitz continuous (\citet{policy_mirror_independent} establishes the conditions under which it is contractive, which we omit here for simplicity).

 \begin{lemma}[Lipschitz continuity of $\Gamma_{\eta}$
 ]\label{learning_operator_lipschitz_lemma}
 For any $\eta > 0$, the operator $\Gamma_{\eta} : \Pi \rightarrow \Pi$ is Lipschitz with constant $L_{\Gamma_{\eta}}$ on ($\Pi,||\cdot||_1$).
\end{lemma}

\subsubsection{Conditions when learning online from samples collected along a single run with $N$ agents}

General theoretical guarantees on online learning from a single run require mixing conditions on the samples along the path. As per \citet{policy_mirror_independent}, we decompose these into the two assumptions below.

The first of these assumptions presumes that throughout training along the single system run, the regulariser $h$ ensures that policies continue taking each action in each state with probability bounded away from zero. \citet{policy_mirror_independent} provides conditions on $h$ that are sufficient for Assumption \ref{Persistence_excitation} to hold; in practice it is achieved for a large class of strongly concave $h$, including those implemented as entropy regularisation, which is how we implement $h$ for our experiments.

\begin{assumption}[Persistence of excitation
]\label{Persistence_excitation}
We assume there exists $p_{inf} > 0$ such that:
\begin{enumerate}
    \item $\pi_{\mathrm{max}}(a|s) \geq p_{inf}$ $\forall s \in \mathcal{S},a \in \mathcal{A}$,
    \item For any $\pi \in \Pi$ and $q \in \mathcal{Q}$ that satisfy, $\forall(s,a) \in \mathcal{S} \times \mathcal{A}$, $\pi(a|s) \geq p_{inf}$ and $0 \leq q(s,a) \leq Q_{\mathrm{max}}$, it holds that $\Gamma^{md}_{\eta}(q,\pi)(a|s) \geq p_{inf}$, $\forall(s,a) \in \mathcal{S} \times \mathcal{A}$.  
\end{enumerate} 
\end{assumption}

The next assumption presumes that, for a persistently excited policy, all agents have a uniformly positive probability of being in any state after a finite time.

\begin{assumption}[Sufficient mixing
]\label{Sufficient_mixing}
For any $\pi\in\Pi$ satisfying $\pi(a|s)\geq p_{inf}>0$ $\forall s\in\mathcal{S},a\in\mathcal{A}$, and any initial states \{$s^i_0$\}$_i \in \mathcal{S}^N$, there exist $T_{mix} > 0, \delta_{mix} > 0$ such that $\mathbb{P}(s^j_{T_{mix}} = s'|$\{$s^i_0$\}$_i)$ $\geq\delta_{mix}$, $\forall s' \in \mathcal{S},j\in[N]$. 
\end{assumption}

\section{Learning with networked, decentralised agents}\label{algorithm_section}

\textbf{Roadmap} We first introduce theoretical versions of our operators and algorithm (Secs. \ref{core_learning}, \ref{communication_section}), in order to show that {our networked framework has sample guarantees bounded between those of the centralised- and independent-learning cases} (Sec. \ref{properties}). We then show that our novel incorporation of an experience replay buffer (Sec. \ref{buffer_section}), along with networked communication, means that empirically we can {remove many of the theoretical assumptions and practically infeasible hyperparameter choices} that are required by the sample guarantees of the theoretical algorithms. In such cases we demonstrate experimentally that {our modified networked algorithm  still respects the theoretical guarantees: it can significantly outperform the independent algorithm, often performing similarly to the central-agent one} (Sec. \ref{experiments}).

\subsection{Learning with $N$ agents from a single run}\label{core_learning}

We begin by outlining the basic procedure for solving the MFG using the $N$-agent empirical distribution and a single, non-episodic system run. The two underlying learning operators are the same for the centralised, independent and networked architectures; in the latter two cases all agents apply the operators individually, while in the centralised setting a single representative agent (the agent with arbitrary index $i=1$) estimates the Q-function and computes an updated policy that is pushed to all the other agents.

Learning agents use the stochastic temporal difference (TD)-learning operator to repeatedly update an estimate of the Q-function of their current policy with respect to the current empirical distribution, i.e. to approximate the operator $\Gamma_q$ (Def. \ref{gamma_q_operator}, Sec. \ref{Further_definitions_and_assumptions}):

\begin{definition}[Stochastic TD-learning operator, simplified from Def. 4.1 in \citet{policy_mirror_independent}]\label{CTD_operator} 
We define $\mathcal{Z} := \mathcal{S} \times \mathcal{A} \times [0,1] \times \mathcal{S} \times\mathcal{A}$, and say that $\zeta^i_t$ is the transition observed by agent $i$ at time $t$, given by $\zeta^i_t = (s^i_t,a^i_t,r^i_t,s^i_{t+1},a^i_{t+1})$. The TD-learning operator $\tilde{F}^{\pi}_{\beta} : \mathcal{Q} \times \mathcal{Z} \rightarrow \mathcal{Q}$ is defined, 
for any $Q \in \mathcal{Q},\zeta_t \in\mathcal{Z},\beta \in \mathbb{R}$, as \[
    \tilde{F}^\pi_{\beta}(Q,\zeta_t) =  Q(s_t,a_t) - \beta\Bigl(Q(s_t,a_t) - r_t - h(\pi(s_t))- \gamma  Q(s_{t+1},a_{t+1})\Bigr).
\]
\end{definition} 

Having estimated the Q-function of their current policy, agents use the Q-function to update this policy via the PMA operator from Def. \ref{PMA_operator}.

The theoretical learning algorithm 
has three nested loops (see Lines \ref{k_loop_theoretical}, \ref{mpg_loop_theoretical} and \ref{mtd_loop_theoretical} of Alg. \ref{networked_algorithm}). The policy update is applied $K$ times. Before the policy update in each of the $K$ loops, agents update their estimate of the Q-function by applying the stochastic TD-learning operator $M_{pg}$ times. Prior to the TD update in each of the $M_{pg}$ loops, agents take $M_{td}$ steps in the environment without updating. The $M_{td}$ loops exist to create a delay between each TD update to reduce bias when using the empirical distribution to approximate the mean field in a non-episodic system run \citep{Kotsalis}. However, we find in our experiments that we are able to essentially remove the inner $M_{td}$ loops (Sec. \ref{discussion}).

\subsection{Decentralised communication between agents}\label{communication_section}

\begin{algorithm}[t]
    \caption{Networked learning with single system run}\label{networked_algorithm}
    \begin{algorithmic}[1]
    \REQUIRE loop parameters $K,M_{pg},M_{td},C$, learning parameters
    $\eta,\{\beta_m\}_{m \in \{0,\dots,M_{pg}-1\}}$, $\lambda, \gamma$, $\{\tau_k\}_{k \in \{0,\dots,K-1\}}$
    
    \REQUIRE initial states \{$s^i_0$\}$_{i=1}^{N}$
    \STATE Set $\pi^i_0 = \pi_{\mathrm{max}}, \forall i$ and $t \leftarrow 0$ \label{init_policy_theoretical}
    \FOR{$k = 0,\dots,K-1$}\label{k_loop_theoretical}
        \STATE $\forall s,a,i : \hat{Q}^i_0 (s,a) = Q_{\mathrm{max}}$ \label{begin_core_theoretical}
        \FOR{$m = 0,\dots,M_{pg}-1$}\label{mpg_loop_theoretical}
            \FOR{$M_{td}$ iterations}\label{mtd_loop_theoretical}
                \STATE Take step $\forall i : a^i_t \sim \pi^i_{k}(\cdot|s^i_t), r^i_{t} = R(s^i_{t},a^i_{t},\hat{\mu}_t),s^i_{t+1} \sim P(\cdot|s^i_{t},a^i_{t},\hat{\mu}_t)$; $t \leftarrow t + 1$
            \ENDFOR
            \STATE Compute TD update ($\forall i$): $\hat{Q}^i_{m+1} = \tilde{F}^{\pi^i_k}_{\beta_m}(\hat{Q}^i_{m},\zeta^i_{t-2})$ (Def. \ref{CTD_operator})
        \ENDFOR \label{end_Qlearning_theoretical}
        
        \STATE PMA step $\forall i : \pi^i_{k+1} = \Gamma^{md}_\eta(\hat{Q}^i_{M_{pg}},\pi^i_k)$ (Def. \ref{PMA_operator}) \label{end_core_theoretical}
        \STATE $\forall i :$ Generate $\sigma^i_{k+1}$ associated with $\pi^{i}_{k+1}$ \label{start_comm_theoretical}
        
        \FOR{$C$ rounds}
            \STATE $\forall i :$ Broadcast $\sigma^i_{k+1}, \pi^{i}_{k+1}$\label{broadcast_theoretical}
            
            \STATE $\forall i : J^i_t = i \cup \{j \in \mathcal{N} : (i,j) \in \mathcal{E}_{t}$\} \label{neighbours_theoretical}
            
            \STATE $\forall i:$ Select $\mathrm{adopted}^i \sim$ Pr$\left(\mathrm{adopted}^i = j\right)$ = $\frac{\exp{(\sigma^j_{k+1}}/\tau_k)}{\sum_{x\in J^i_t}\exp{(\sigma^x_{k+1}}/\tau_k)}$ $\forall j \in J^i_t$  \label{select_theoretical}
                    
            \STATE $\forall i : \sigma^i_{k+1} \leftarrow \sigma^{\mathrm{adopted}^i}_{k+1}, \pi^i_{k+1} \leftarrow \pi^{\mathrm{adopted}^i}_{k+1}$\label{adopt_theoretical}
            
            \STATE Take step $\forall i : a^i_t \sim \pi^i_{k+1}(\cdot|s^i_t), r^i_{t} = R(s^i_{t},a^i_{t},\hat{\mu}_t),s^i_{t+1} \sim P(\cdot|s^i_{t},a^i_{t},\hat{\mu}_t)$; $t \leftarrow t + 1$\label{stepInComm_theoretical}
            
            \ENDFOR\label{end_comm_theoretical}

    \ENDFOR
    \STATE \textbf{return} policies \{$\pi^i_K$\}$_{i=1}^{N}$
    \end{algorithmic}
\end{algorithm}

In our {novel algorithm Alg. \ref{networked_algorithm}}, agents compute policy updates in a decentralised way as in the independent case (Lines \ref{begin_core_theoretical}-\ref{end_core_theoretical}), before exchanging policies with neighbours in Lines \ref{start_comm_theoretical}-\ref{end_comm_theoretical} by the following method, which allows policies to spread through the population.\footnote{As discussed in Sec. \ref{related_work}, our communication method is reminiscent of the use of fitness functions in distributed evolutionary algorithms \citep{Eiben2015,survivability}.} Coupled to their updated policy $\pi^{i}_{k+1}$, agents generate a scalar value $\sigma^{i}_{k+1}$ (Line \ref{start_comm_theoretical}). The value provides information that helps agents decide between policies that they may wish to adopt from neighbours. Different methods for choosing between values received from neighbours, 
and for generating the values in the first place, lead to different policies spreading through the population. For example, generating or choosing $\sigma_{k+1}^{i}$ at random leads to policies being exchanged at random (required in Thm. \ref{random_adoption_theorem_full}), whereas generating $\sigma_{k+1}^{i}$ as an approximation of the return of $\pi_{k+1}^{i}$ and then selecting the highest received value of $\sigma_{k+1}^{j}$ leads to better performing policies spreading through the population. The latter is the approach we use for accelerating learning empirically (described in Sec. \ref{practical_generation} on the practical running of our algorithm), albeit we use a softmax rather than a max function for selecting between received values. However, for generality in our theoretical results, we do not focus on a specific method for generating $\sigma^{i}_{k+1}$, such that it can be arbitrary for Thms. \ref{random_adoption_theorem_full} and \ref{Consecutive_stability_guarantee_theorem} below, and with few restrictions for Thms. \ref{non_random_adoption_theorem_full} and \ref{order_of_difference_theorem}. 

Agents broadcast their policy $\pi^{i}_{k+1}$ and the associated $\sigma^{i}_{k+1}$ value to their neighbours (Line \ref{broadcast_theoretical}). Agents have a certain broadcast radius, defining the structure of the possibly time-varying communication network. Of the policies and associated values received by a given agent (including its own) (Line \ref{neighbours_theoretical}), the agent selects a $\sigma^{j}_{k+1}$ with a probability defined by a softmax function over the received values, and \textit{adopts} the policy associated with this $\sigma^{i}_{k+1}$, i.e. it sets its own current $\pi^{i}_{k+1}$ and $\sigma^{i}_{k+1}$ to the ones it has selected (Lines \ref{select_theoretical}, \ref{adopt_theoretical}). This process repeats for $C$ communication rounds, before the Q-function estimation steps begin again. After each communication round, the agents take a step in the environment (Line \ref{stepInComm_theoretical}), such that if the communication network is affected by the agents' states, then agents that are unconnected from any others in a given communication round might become connected in the next. ({In our experiments we set $C$ as 1 to show the benefits to convergence speed brought by even a single communication round}.) We assume the softmax function is subject to a possibly time-varying temperature parameter $\tau_k$. We discuss the effects of the values of $C$ and $\tau_k$, and the mechanism for generating $\sigma^{i}_{k+1}$, in subsequent sections.

\begin{remark}\label{generalisation_remark}
    Our networked architecture is effectively a generalisation of both the central-agent and independent settings (Algs. 2, 3, \citet{policy_mirror_independent}). The independent setting is the special case where there is no communication, i.e. $C = 0$ - this serves as an implicit ablation of our communication scheme. The central-agent setting is the special case when $\sigma^i_{k+1}$ is generated from a unique ID for each agent, with the central learner agent assumed to generate the highest value by default. In this case we assume $\tau_k \rightarrow 0$ (such that the softmax becomes a max function), and that the communication network becomes jointly connected repeatedly, so the central learner's policy is always adopted by the entire population, assuming $C$ is large enough that the number of jointly connected collections of graphs occurring within $C$ is equal to the largest diameter of the union of any collection \citep{Distributed_Averaging,average_belief}.  
\end{remark}

\begin{remark}\label{centralised_version}
    In practice, when referring to a central-agent version of the networked Alg. \ref{networked_algorithm}, for simplicity we assume there is no networked communication and instead that the updated policy $\pi^1_{k+1}$ of the representative learner $i = 1$ is pushed to all agents after Line \ref{end_core_theoretical}, as in Alg. 2 of \citep{policy_mirror_independent}.  
\end{remark}

\section{Theoretical results}\label{properties}

In this section we first give two theoretical results comparing the sample guarantees of our networked case with those of the other settings; the results respectively depend on whether the networked agents select which communicated policies to adopt at random or not. We then provide the order of the difference in these bounds in the non-random case in terms of the network structure and number of communication rounds. We finally give 
a policy-update stability guarantee, which applies in all scenarios. 

We begin by recalling the sample guarantees of the architecture where agents learn entirely independently \citep{policy_mirror_independent}.

\begin{lemma}[Independent learning, from Thm. 4.5, \citet{policy_mirror_independent}]\label{independent_learning_theorem}
For $p_{inf}$ and $\delta_{mix}$ defined in Assumptions \ref{Persistence_excitation} and \ref{Sufficient_mixing} respectively, define $t_0 := \frac{16(1+\gamma)^2}{((1 - \gamma)\delta_{mix}p_{inf})^2}$.
Assume that Assumptions \ref{lip_cont_assumpt}, \ref{stable_population}, \ref{Persistence_excitation} and \ref{Sufficient_mixing} hold, and that $\pi^*$ is the unique MFG-NE policy. For $L_{\Gamma_{\eta}}$ defined in Lem. \ref{learning_operator_lipschitz_lemma}, we assume $\eta > 0$ satisfies $L_{\Gamma_{\eta}} < 1$. The learning rates are $\beta_m = \frac{2}{(1 - \gamma)(t_0 + m -1)}$ $\forall m \geq 0$, and let $\varepsilon > 0$ be arbitrary. There exists a problem-dependent constant $a \in [0,\infty)$ such that if $K =\frac{\log8\varepsilon^{-1}}{\log L^{-1}_{\Gamma_\eta}}$, $M_{pg} > \mathcal{O}(\varepsilon^{-2-a})$ and $M_{td} > \mathcal{O}(log^{2}\varepsilon^{-1})$, then the random output \{$\pi^{i}_{K}$\}$_{i}$ of Alg. \ref{networked_algorithm} when run with $C = 0$ (such that there is no communication) satisfies for all agents $i \in \{1,\dots,N\}$,\[\mathbb{E}\left[||\pi^{i}_K - \pi^*||_1\right] \; \leq \; \varepsilon + \mathcal{O}\left(\frac{1}{\sqrt{N}}\right).\] 
\end{lemma}


We first give a result for the trivial situation of random adoption to provide an intuition that networked communication preserves the sample guarantees of independent learning, before showing the conditions under which the latter can be outperformed. 

\begin{theorem}[Networked learning with random adoption]\label{random_adoption_theorem_full}
For $p_{inf}$ and $\delta_{mix}$ defined in Assumptions \ref{Persistence_excitation} and \ref{Sufficient_mixing} respectively, define $t_0 := \frac{16(1+\gamma)^2}{((1 - \gamma)\delta_{mix}p_{inf})^2}$.
Assume that Assumptions \ref{lip_cont_assumpt}, \ref{stable_population}, \ref{Persistence_excitation} and \ref{Sufficient_mixing} hold, and that $\pi^*$ is the unique MFG-NE policy. For $L_{\Gamma_{\eta}}$ defined in Lem. \ref{learning_operator_lipschitz_lemma}, we assume $\eta > 0$ satisfies $L_{\Gamma_{\eta}} < 1$. The learning rates are $\beta_m = \frac{2}{(1 - \gamma)(t_0 + m -1)}$ $\forall m \geq 0$, and let $\varepsilon > 0$ be arbitrary.  Let us set $C > 0$ and $\tau_k \rightarrow \infty$. There exists a problem-dependent constant $a \in [0,\infty)$ such that if $K =\frac{\log8\varepsilon^{-1}}{\log L^{-1}_{\Gamma_\eta}}$, $M_{pg} > \mathcal{O}(\varepsilon^{-2-a})$ and $M_{td} > \mathcal{O}(log^{2}\varepsilon^{-1})$, then the random output \{$\pi^{i}_{K}$\}$_{i}$ of {Alg. \ref{networked_algorithm} preserves the sample guarantees of the independent-learning case} given in Lem. \ref{independent_learning_theorem}, i.e. the output satisfies, for all agents $i \in \{1,\dots,N\}$,
\[\mathbb{E}\left[||\pi^{i}_K - \pi^*||_1\right] \; \leq \; \varepsilon + \mathcal{O}\left(\frac{1}{\sqrt{N}}\right).  \] 
\end{theorem}


\begin{proof}
If $\tau_k \rightarrow \infty$, the softmax function that defines the probability of a received policy being adopted in Line \ref{select_theoretical} of Alg. \ref{networked_algorithm} gives a uniform distribution. Policies are thus exchanged at random between communicating agents for an arbitrary $C > 0$ rounds, which does not affect the random output of the algorithm, such that the random output satisfies the same expectation as if $C = 0$. \end{proof}

If $\sigma^{i}_{k+1}$ is generated arbitrarily and uniquely for each $i$, then for $\tau_k \in \mathbb{R}_{>0}$ (such that the softmax function gives a non-uniform distribution and adoption of received policies is therefore non-random), {the sample complexity of the networked algorithm is bounded between that of the centralised and independent algorithms}:

\begin{theorem}[Networked learning with non-random adoption]\label{non_random_adoption_theorem_full}
Assume that Assumptions \ref{lip_cont_assumpt}, \ref{stable_population}, \ref{Persistence_excitation} and \ref{Sufficient_mixing} hold, and that Alg. \ref{networked_algorithm} is run with learning rates and constants as defined in Thm. \ref{random_adoption_theorem_full}, except now let us set $\tau_k \in \mathbb{R}_{>0}$. Assume that $\sigma^{i}_{k+1}$ is generated uniquely for each $i$, in a manner independent of any metric related to $\pi^{i}_{k+1}$, e.g. $\sigma^{i}_{k+1}$ is random or related only to the index $i$ (so as not to bias the spread of any particular policy). Let the random output of this Algorithm be denoted as \{$\pi^{i,net}_{K}$\}$_{i}$. Also consider an independent-learning version of the algorithm (i.e. with the same parameters except $C = 0$) and denote its random output \{$\pi^{i,ind}_{K}$\}$_{i}$; and a central-agent version of the algorithm with the same parameters (see Rem. \ref{centralised_version}) and denote its random output as $\pi^{cent}_{K}$. Then for all agents $i \in \{1,\dots,N\}$, the random outputs \{$\pi^{i,net}_K$\}$_i$, \{$\pi^{i,ind}_K$\}$_i$ and $\pi^{cent}_{K}$ satisfy the following relations, where $ub_{net}$, $ub_{ind}$ and $ub_{cent}$ are respective upper bounds for each case: 
\[\mathbb{E}\left[||\pi^{cent}_K - \pi^*||_1\right]   \leq  ub_{cent}, \;\;\;\;\;\mathbb{E}\left[||\pi^{i,net}_K - \pi^*||_1\right] \leq  ub_{net},\;\;\;\;\;\mathbb{E}\left[||\pi^{i,ind}_K - \pi^*||_1\right] \leq ub_{ind},\]
\[\text{where } \;\;\;\; ub_{cent}  \; \leq \; ub_{net} \leq \; ub_{ind} \; = \; \varepsilon + \mathcal{O}\left(\frac{1}{\sqrt{N}}\right).\]
\end{theorem}


\begin{proof}

We build off the proof of our Lem. \ref{independent_learning_theorem}, given in Thm. D.9 of \citet{policy_mirror_independent}. There the sample guarantees of the independent case are worse than those of the centralised algorithm as a result of the divergence between the decentralised policies due to the stochasticity of the PMA updates. For an arbitrary policy $\bar{\pi}_k\in\Pi$, for all $k = 0,1,\dots,K$ define the policy divergence as the random variable $\Delta_{k}$ := $\sum_{i=1}^N ||\pi^i_k -\bar{\pi}_k||_1$. We can say that $\Delta_{k,cent} = 0$ $\forall k$ is the divergence in the central-agent case, while in the networked case the policy divergence is $\Delta_{k+1,c}$ after communication round $c \in 1,\dots,C$. The independent case is equivalent to the scenario when $C = 0$, such that its policy divergence can be written $\Delta_{k+1,0}$.

For $\tau_k \in \mathbb{R}_{>0}$, the adoption probability Pr$\left(\mathrm{adopted}^i = \sigma^j_{k+1}\right)$ = $\frac{\exp{(\sigma^j_{k+1}}/\tau_k)}{\sum^{[J^i_t]}_{x=1}\exp{(\sigma^x_{k+1}}/\tau_k)}$ (as in Line \ref{select_theoretical} of Alg. \ref{networked_algorithm}) is higher for some $j \in J^i_t$ than for others. This means that for $c > 0$ for which there are communication links in the population, in expectation the number of unique policies in the population will decrease, as it will likely become that $\pi^{i}_{k+1} = \pi^{j}_{k+1}$ for some $i,j\in \{1,\dots,N\}$. As such, $\Delta_{k+1,cent} \leq \mathbb{E}\left[\Delta_{k+1,C}\right] \leq \mathbb{E}\left[\Delta_{k+1,0}\right]$, i.e. the policy divergence in the independent-learning case is expected to be greater than or equal to that of the networked case. 

The proof of Lem. \ref{independent_learning_theorem} given in Thm. D.9 of \citet{policy_mirror_independent} ends with, for constants $\chi$ and $\xi$, 
    \[\mathbb{E}\left[||\pi^i_K - \pi^*||_1\right] \leq 2L_{\Gamma_{\eta}}^K + \frac{\chi}{1-L_{\Gamma_{\eta}}} + \xi\sum^{K-1}_{k=1}L_{\Gamma_{\eta}}^{K-k-1}\mathbb{E}\left[\Delta_{k}\right],\]
where in our context the policy divergence in the independent case $\mathbb{E}\left[\Delta_{k+1}\right]$ is equivalent to $\mathbb{E}\left[\Delta_{k+1,C}\right]$ when $C = 0$, i.e. $\mathbb{E}\left[\Delta_{k+1,0}\right]$.

Thus, for all agents $i \in \{1,\dots,N\}$, the random outputs \{$\pi^{i,net}_K$\}$_i$, \{$\pi^{i,ind}_K$\}$_i$ and $\pi^{cent}_{K}$ satisfy: 
 \[\mathbb{E}\left[||\pi^{i,ind}_K - \pi^*||_1\right] \leq ub_{ind} = 2L_{\Gamma_{\eta}}^K + \frac{\chi}{1-L_{\Gamma_{\eta}}} + \xi\sum^{K-1}_{k=1}L_{\Gamma_{\eta}}^{K-k-1}\mathbb{E}\left[\Delta_{k,0}\right],\]
 \[\mathbb{E}\left[||\pi^{i,net}_K - \pi^*||_1\right] \leq ub_{net} = 2L_{\Gamma_{\eta}}^K + \frac{\chi}{1-L_{\Gamma_{\eta}}} + \xi\sum^{K-1}_{k=1}L_{\Gamma_{\eta}}^{K-k-1}\mathbb{E}\left[\Delta_{k,C}\right],\]
 \[\;\;\;\;\;\mathbb{E}\left[||\pi^{cent}_K - \pi^*||_1\right] \leq ub_{cent} = 2L_{\Gamma_{\eta}}^K + \frac{\chi}{1-L_{\Gamma_{\eta}}} + \xi\sum^{K-1}_{k=1}L_{\Gamma_{\eta}}^{K-k-1}\mathbb{E}\left[\Delta_{k,cent}\right].\]

Since $\Delta_{k+1,cent} \leq \mathbb{E}\left[\Delta_{k+1,C}\right] \leq \mathbb{E}\left[\Delta_{k+1,0}\right]$, we obtain our result, i.e. 
\[ub_{cent}  \; \leq \; ub_{net} \leq \; ub_{ind} \; = \; \varepsilon + \mathcal{O}\left(\frac{1}{\sqrt{N}}\right).\]\end{proof}


We recall the following lemma from \citet{policy_mirror_independent}; we use it in Rem. \ref{qbias} to aid intuitive understanding of the result given in Thm. \ref{non_random_adoption_theorem_full}.

\begin{lemma}[Conditional TD learning from a single continuous run of the empirical distribution of $N$ agents, from Thm. 4.2, \citet{policy_mirror_independent}]\label{ctd_theorem}
Define $t_0 := \frac{16(1+\gamma)^2}{((1 - \gamma)\delta_{mix}p_{inf})^2}$. Assume that Assumption \ref{Sufficient_mixing} holds and let policies \{$\pi^i$\}$_i$ be given such that $\pi^i(a|s) \geq p_{inf}$ $\forall i$. Assume Lines \ref{begin_core_theoretical}-\ref{end_Qlearning_theoretical} of Alg. \ref{networked_algorithm} are run with policies \{$\pi^i$\}$_i$, arbitrary initial agents states \{$s^i_0$\}$_i$, learning rates $\beta_m = \frac{2}{(1 - \gamma)(t_0 + m -1)}$, $\forall m \geq 0$ and $M_{pg} > \mathcal{O}(\varepsilon^{-2})$, $M_{td} > \mathcal{O}(\log \varepsilon^{-1})$. If $\bar{\pi}\in\Pi$ is an arbitrary policy, $\Delta$ := $\sum_{i=1}^N ||\pi^i -\bar{\pi}||_1$ and $Q^*$ := $Q_{h}(\cdot,\cdot|\bar{\pi},\mu_{\bar{\pi}})$, then the random output $\hat{Q}^i_{M_{pg}}$ of Lines \ref{begin_core_theoretical}-\ref{end_Qlearning_theoretical} satisfies
\[\mathbb{E}\left[||\hat{Q}^i_{M_{pg}} - \; Q^{*}||_{\infty}\right] \leq \varepsilon + \mathcal{O}\left(\frac{1}{\sqrt{N}} +\frac{1}{N}\Delta + ||\pi^{i} - \bar{\pi}||_{1}\right).\]

\end{lemma}

\begin{remark}\label{qbias}
It may help to see that our Thm. \ref{non_random_adoption_theorem_full} is a consequence of the following. Denote $\hat{Q}^{i,net}_{M_{pg}}$,  $\hat{Q}^{i,ind}_{M_{pg}}$ and $\hat{Q}^{cent}_{M_{pg}}$ as the random outputs of Lines \ref{begin_core_theoretical}-\ref{end_Qlearning_theoretical} of Alg. \ref{networked_algorithm} in the networked, independent and central-agent cases respectively. In Lem. \ref{ctd_theorem}, we can see that policy divergence gives bias terms in the estimation of the Q-value. Therefore, given $\Delta_{k+1,cent} \leq \mathbb{E}\left[\Delta_{k+1,C}\right] \leq \mathbb{E}\left[\Delta_{k+1,0}\right]$, we can also say 
    \[\mathbb{E}\left[||\hat{Q}^{cent}_{M_{pg}} - Q^{*}||_{\infty}\right] \; \leq \; \mathbb{E}\left[||\hat{Q}^{i,net}_{M_{pg}} - Q^{*}||_{\infty}\right] \\
     \leq \; \mathbb{E}\left[||\hat{Q}^{i,ind}_{M_{pg}} - Q^{*}||_{\infty}\right].\]
In other words, the networked case will require the same or fewer outer iterations $K$ to reduce the variance caused by this bias than the independent case requires (where the bias is non-vanishing), and the same or more iterations than the central-agent case requires.
\end{remark}

We now provide a result that shows how the sample guarantees of our networked architecture varies along the spectrum between the central-agent and independent cases depending on the amount of communication that occurs.

\begin{theorem}[Relation between communication network structure and order of difference between the architectures' bounds]\label{order_of_difference_theorem} In addition to the assumptions in Thm. \ref{non_random_adoption_theorem_full}, now also assume that the communication network $\mathcal{G}_t$ remains static and connected during the $C$ communication rounds. Assume also the diameter $d_\mathcal{G}$ of the network is equal for all $k$. Let us set $\tau_k \;\forall k$ as a small positive constant chosen to be sufficiently close to zero that the softmax essentially becomes a max function. Then, for the tight bound big Theta ($\Theta$), we can say that the difference in the upper bounds $ub_{net}$, $ub_{ind}$ and $ub_{cent}$ from Thm. \ref{non_random_adoption_theorem_full} depends on $C$ and the network diameter $d_\mathcal{G}$ as follows (where the `$\approx$' relation comes from the approximate spread of policies through the network as explained in the proof):
\begin{align*}
    ub_{cent} + \Theta\left(f(C, d_\mathcal{G})\right) \;\;\;\; \approx \;\;\;\;  ub_{net} \;\;\;\;  \approx \;\; \;\; ub_{ind} - \Theta\left(1- f(C, d_\mathcal{G})\right),
\end{align*}
for the piecewise function $f(C, d_\mathcal{G})$ defined as 
\[f(C, d_\mathcal{G}) = 
\begin{cases} 
 \left(1 - \frac{1}{d_\mathcal{G}}\right)^C &\text{if } C < d_\mathcal{G},\\  
0 &\text{if } C \geq d_\mathcal{G}
\end{cases}
.\]
When $C \geq d_\mathcal{G}$, $ub_{net} = ub_{cent}$, so for $C > d_\mathcal{G}$ there is no additional improvement over the centralised bound. Equally when $C=0$, we have exactly $ub_{net} = ub_{ind}$.
\end{theorem} 

\begin{proof}
From the proof of Thm. \ref{non_random_adoption_theorem_full} 
we have: 
 \[\mathbb{E}\left[||\pi^{i,ind}_K - \pi^*||_1\right] \leq ub_{ind} = 2L_{\Gamma_{\eta}}^K + \frac{\chi}{1-L_{\Gamma_{\eta}}} + \xi\sum^{K-1}_{k=1}L_{\Gamma_{\eta}}^{K-k-1}\mathbb{E}\left[\Delta_{k,0}\right],\]
 \[\mathbb{E}\left[||\pi^{i,net}_K - \pi^*||_1\right] \leq ub_{net} = 2L_{\Gamma_{\eta}}^K + \frac{\chi}{1-L_{\Gamma_{\eta}}} + \xi\sum^{K-1}_{k=1}L_{\Gamma_{\eta}}^{K-k-1}\mathbb{E}\left[\Delta_{k,C}\right],\]
 \[\;\;\;\;\;\mathbb{E}\left[||\pi^{cent}_K - \pi^*||_1\right] \leq ub_{cent} = 2L_{\Gamma_{\eta}}^K + \frac{\chi}{1-L_{\Gamma_{\eta}}} + \xi\sum^{K-1}_{k=1}L_{\Gamma_{\eta}}^{K-k-1}\mathbb{E}\left[\Delta_{k,cent}\right].\]
 
 Say that $\sigma^{\mathrm{max}}_{k+1}$ is the highest $\sigma^{i}$ value in the population before the communication rounds at ${k+1}$. With a static, connected network and $\tau_k$ close to 0 for all $k$, max-consensus will always be reached on $\sigma^{\mathrm{max}}_{k+1}$ after $C = d_\mathcal{G}$ communication rounds, such that $\Delta_{k,cent} = \Delta_{k,d_\mathcal{G}} = 0$ \citep{5348437}. The convergence rate of the max-consensus algorithm is $\frac{1}{d_\mathcal{G}}$ \citep{5348437}, i.e. there is a decrease in the \textit{number of policies in the population} by a factor of \textit{approximately} $\frac{1}{d_\mathcal{G}}$ with each communication round up to $C = d_\mathcal{G}$, and therefore there is also a decrease in the \textit{policy divergence} $\mathbb{E}\left[\Delta_{k,c}\right]$ by a factor of approximately $\frac{1}{d_\mathcal{G}}$ with each communication round. Thus 
 \[ \mathbb{E}\left[\Delta_{k,c+1}\right] \approx \mathbb{E}\left[\Delta_{k,c}\right] - \left(\mathbb{E}\left[\Delta_{k,c}\right]  \times \frac{1}{d_\mathcal{G}}\right) \text{, simplifying to}\]
\[\mathbb{E}\left[\Delta_{k,c+1}\right] \approx \mathbb{E}\left[\Delta_{k,c}\right] \times \left(1 - \frac{1}{d_\mathcal{G}}\right).\] 
By induction
\[\mathbb{E}\left[\Delta_{k,C}\right] \approx \mathbb{E}\left[\Delta_{k,0}\right] \times \left(\left(1 - \frac{1}{d_\mathcal{G}}\right)^C\right),\]
however, we know that $\Delta_{k,d_\mathcal{G}} = 0$, so we can more accurately use the piecewise function $f(C, d_\mathcal{G})$, defined as:
\[f(C, d_\mathcal{G}) = 
\begin{cases} 
\left(1 - \frac{1}{d_\mathcal{G}}\right)^C &\text{if } C < d_\mathcal{G},\\
0 &\text{if } C \geq d_\mathcal{G}
\end{cases}, 
\]
giving
\[\mathbb{E}\left[\Delta_{k,C}\right] \approx \mathbb{E}\left[\Delta_{k,0}\right] \times f(C, d_\mathcal{G}).\]
 We can therefore also say: 
  \[ub_{ind} = 2L_{\Gamma_{\eta}}^K + \frac{\chi}{1-L_{\Gamma_{\eta}}} + \xi\sum^{K-1}_{k=1}L_{\Gamma_{\eta}}^{K-k-1}\mathbb{E}\left[\Delta_{k,0}\right],\]
 \[ub_{net} \approx 2L_{\Gamma_{\eta}}^K + \frac{\chi}{1-L_{\Gamma_{\eta}}} + \xi\sum^{K-1}_{k=1}L_{\Gamma_{\eta}}^{K-k-1}\mathbb{E}\left[\Delta_{k,0}\right] \times f(C, d_\mathcal{G}),\] 
 \[\;\;\;\;\;ub_{cent} = 2L_{\Gamma_{\eta}}^K + \frac{\chi}{1-L_{\Gamma_{\eta}}}.\]
We therefore firstly have 
\[ub_{ind} - ub_{net} \approx \xi\sum^{K-1}_{k=1}L_{\Gamma_{\eta}}^{K-k-1}\mathbb{E}\left[\Delta_{k,0}\right] - \xi\sum^{K-1}_{k=1}L_{\Gamma_{\eta}}^{K-k-1}\mathbb{E}\left[\Delta_{k,0}\right] \times f(C, d_\mathcal{G}),\]
which simplifies to 
\[ub_{ind} - ub_{net} \approx \xi\sum^{K-1}_{k=1}L_{\Gamma_{\eta}}^{K-k-1}\mathbb{E}\left[\Delta_{k,0}\right] \times \left(1 - f(C, d_\mathcal{G})\right).\]
This gives us one of the results, where we focus on the functional dependence on $C$ and $d_\mathcal{G}$ by using the tight bound big Theta ($\Theta$): 
\[ub_{net} \approx ub_{ind} - \Theta\left(1 - f(C, d_\mathcal{G})\right).\]
Secondly, we have
\[ub_{net} \approx ub_{cent} + \xi\sum^{K-1}_{k=1}L_{\Gamma_{\eta}}^{K-k-1}\mathbb{E}\left[\Delta_{k,0}\right] \times f(C, d_\mathcal{G}),\]
 giving us the second result 
\[ub_{net} \approx ub_{cent} + \Theta\left(f(C, d_\mathcal{G})\right).\]
\end{proof}

\begin{remark}
    If it is always $\sigma^1_{k+1}$ and $\pi^{1}_{k+1}$ that is adopted by the whole population (i.e. $i=1$), then this is exactly the same as the central-agent case. If the $\sigma^j_{k+1}$ and $\pi^{j}_{k+1}$ that gets adopted has different $j$ for each $k$, then this is akin to a version of the central-agent setting where the index of the representative learning agent may differ for each $k$.
\end{remark}

\begin{remark}\label{dynamic_graph_remark}
    Thm. \ref{order_of_difference_theorem} depends on the assumptions that the communication network is static and fixed, and has the same diameter $d_\mathcal{G}$ for all $k$. If we instead only assume that the network repeatedly becomes jointly connected, we can replace $d_\mathcal{G}$ in the results in Thm. \ref{order_of_difference_theorem} with $d_{avg}\cdot\omega$, namely the average diameter of the union of each jointly connected collection of graphs multiplied by the average number $\omega$ of graphs in each jointly connected collection. As noted in Rem. \ref{generalisation_remark}, max-consensus is reached if $C$ is large enough that the number of jointly connected collections of graphs occurring within $C$ is equal to the largest diameter of the union of any collection. This is equivalent to the central-agent case; there is no added benefit to higher values of $C$ than this. 
\end{remark}

\begin{remark}\label{non_static_tau_remark}
    Thm. \ref{order_of_difference_theorem} assumes $\tau_k$ is a small positive value close to 0 such that the softmax function becomes a max function. If we assume instead $\tau_k \in \mathbb{R}_{>0}$ is not close to 0 such that the softmax function is less peaked, then we have $ub_{net}$ $\rightarrow$ $ub_{ind}$ as $C \rightarrow 0$, and   $ub_{net} \rightarrow$ $ub_{cent}$ as $C \rightarrow \infty$. This is because the spread of policies is now probabilistic rather than deterministic, and depends on the interplay of $\tau_k$ with how large are the differences in the received values of $\sigma^j_{k+1}$. Therefore consensus (and hence reduction in divergence between policies) is reached only asymptotically. This applies to both static, connected networks and to repeatedly jointly connected ones, assuming the latter becomes jointly connected infinitely often.
\end{remark}

For completeness, we finally give a stability guarantee that follows from the earlier theorems.

\begin{theorem}[Policy-update stability guarantee]\label{Consecutive_stability_guarantee_theorem}
Let Alg. \ref{networked_algorithm} run as per Thm. \ref{random_adoption_theorem_full} or Thms. \ref{non_random_adoption_theorem_full}/\ref{order_of_difference_theorem}, and say that $\varepsilon_k$ is the error term at iteration $k =\frac{\log8\varepsilon_k^{-1}}{\log L^{-1}_{\Gamma_\eta}}$. For all agents $i$, the maximum possible distance between $\pi^{i,net}_{k}$ and $\pi^{i,net}_{k+1}$ is given by $\mathbb{E}\left[||\pi^{i,net}_k - \pi^{i,net}_{k+1}||_1\right]  \leq  \varepsilon_{k} + \varepsilon_{k+1} + \mathcal{O}\left(\frac{1}{\sqrt{N}}\right)$. This bound provides a stability guarantee during the learning process; moreover the bound shrinks with each successive $k$ since $\varepsilon_{k}$ decreases with $k$. Equivalent analysis can also be conducted for both the centralised and independent cases.\end{theorem} 

\begin{proof} Thms. \ref{random_adoption_theorem_full}, \ref{non_random_adoption_theorem_full} and \ref{order_of_difference_theorem} bound the difference between each agent’s current policy $\pi^i_{k}$ and the unique equilibrium policy $\pi^*$, with the difference depending on the bias term $\varepsilon_k$ that relates to the iteration $k$ as indicated. Policies $\pi^i_{k}$ and $\pi^i_{k+1}$ fall within balls centred on $\pi^*$ with radii of $\varepsilon_{k} + \mathcal{O}\left(\frac{1}{\sqrt{N}}\right)$ and $\varepsilon_{k+1} + \mathcal{O}\left(\frac{1}{\sqrt{N}}\right)$ respectively. This means that the maximum possible distance between $\pi^i_{k}$ and $\pi^i_{k+1}$ is the sum of these radii, i.e. $\mathbb{E}\left[||\pi^{i}_k - \pi^i_{k+1}||_1\right]  \leq  \varepsilon_{k} + \varepsilon_{k+1} + \mathcal{O}\left(\frac{1}{\sqrt{N}}\right)$, giving the result.\end{proof}

\section{Practical modifications to theoretical algorithms for empirical use}\label{practical_enhancements}


The theoretical analysis in Sec. \ref{properties} 
requires algorithmic hyperparameters (see Thm. \ref{random_adoption_theorem_full}) that render convergence impractically slow in all of the central-agent, independent and networked cases.  {In particular, the values of $\delta_{mix}$ and $p_{inf}$ give rise to very large $t_0$, causing very small learning rates $\{\beta_m\}_{m \in \{0,\dots,M_{pg}-1\}}$, and necessitating very large values for $M_{td}$ and $M_{pg}$
.} Indeed \citet{policy_mirror_independent} do not provide empirical demonstrations of their algorithms for the central-agent and independent cases. Our experiments in Sec. \ref{ablation_buffer_experiment_section} show that with these algorithms, none of the architectures appear to improve their returns at all without extremely high numbers of inner loops that would take impractically long to run on standard computers, taking many days or even many weeks. 

For convergence of the algorithms in practical time, we seek to drastically increase $\{\beta_m\}_m$ and reduce $M_{td}$ and $M_{pg}$. We found empirically that the two algorithmic enhancements below helped achieve feasible convergence times with significantly reduced numbers of loops. The first involves recycling transitions using a buffer, and the second gives a principled way of selecting $\sigma^{i}_{k+1}$ in Line \ref{start_comm_theoretical} in Alg. \ref{networked_algorithm}.

\begin{algorithm}[t]
    \caption{Networked learning with experience replay and performance-related generation of $\sigma^{i}_{k+1}$} \label{networked_algorithm_experience_replay}
    \begin{algorithmic}[1]
    \REQUIRE loop parameters $K, M_{pg}, M_{td}, C$, \textcolor{blue}{$L$}, \textcolor{orange}{E}, learning parameters
    $\eta$, \textcolor{blue}{$\beta$}, $\lambda, \gamma$, $\{\tau_k\}_{k \in \{0,\dots,K-1\}}$
    
    \REQUIRE initial states \{$s^i_0$\}$_{i=1}^{N}$
    \STATE Set $\pi^i_0 = \pi_{\mathrm{max}}, \forall i$ and $t \leftarrow 0$ \label{init_policy}
    \FOR{$k = 0,\dots,K-1$}
        \STATE $\forall s,a,i : \hat{Q}^i_0 (s,a) = Q_{\mathrm{max}}$ 
        \color{blue}
        \STATE $\forall i$: Empty $i$'s buffer \label{empty_buffer}
        \color{black}
        \FOR{$m = 0,\dots,M_{pg}-1$}\label{mpg_loop}
            \FOR{$M_{td}$ iterations}\label{mtd_loop}
                \STATE Take step $\forall i : a^i_t \sim \pi^i_{k}(\cdot|s^i_t), r^i_{t} = R(s^i_{t},a^i_{t},\hat{\mu}_t),s^i_{t+1} \sim P(\cdot|s^i_{t},a^i_{t},\hat{\mu}_t)$; $t \leftarrow t + 1$
            \ENDFOR
            \color{blue}
            \STATE $\forall i$: Add $\zeta^i_{t-2}$ to $i$'s buffer \label{store_transition}
            \color{black}
        \ENDFOR 

        \color{blue}
        \FOR{$l = 0,\dots,L-1$}\label{beginlearningwithbuffer}
            \STATE $\forall i : $ Shuffle buffer
            \FOR{transition $\zeta^i_b$ in $i$'s buffer ($\forall i$)} 
                \STATE Compute TD update ($\forall i$): $\hat{Q}^i_{m+1} = \tilde{F}^{\pi^i_k}_{\beta}(\hat{Q}^i_{m},\zeta^i_{t-2})$ (see Def. \ref{CTD_operator})
            \ENDFOR
        \ENDFOR \label{endlearningwithbuffer}
        \color{black}
        
        \STATE PMA step $\forall i : \pi^i_{k+1} = \Gamma^{md}_\eta(\hat{Q}^i_{M_{pg}},\pi^i_k)$ (see Def. \ref{PMA_operator}) \label{pma_line_replay_algorithm}
        \color{orange}\STATE $\forall i :$ $\sigma^i_{k+1} \gets 0$
        \FOR{$e = 0,\dots,E-1$ evaluation steps}\label{evaluation_begin}
        \STATE Take step $\forall i : a^i_t \sim \pi^i_{k+1}(\cdot|s^i_t), r^i_{t} = R(s^i_{t},a^i_{t},\hat{\mu}_t),s^i_{t+1} \sim P(\cdot|s^i_{t},a^i_{t},\hat{\mu}_t)$
        \STATE $\forall i :$ $\sigma^i_{k+1} \gets \sigma^i_{k+1} + \gamma^{e}(r^i_{t}+h(\pi^i_{k+1}(s^i_t)))$ 
        \STATE $t \leftarrow t + 1$
        \ENDFOR\label{evaluation_end}
        \color{black}
        \FOR{$C$ rounds}
            \STATE $\forall i :$ Broadcast $\sigma^i_{k+1}, \pi^{i}_{k+1}$
            
            \STATE $\forall i : J^i_t = i \cup \{j \in \mathcal{N} : (i,j) \in \mathcal{E}_{t}$\} 
            
           \STATE $\forall i:$ Select $\mathrm{adopted}^i \sim$ Pr$\left(\mathrm{adopted}^i = j\right)$ = $\frac{\exp{(\sigma^j_{k+1}}/\tau_k)}{\sum_{x\in J^i_t}\exp{(\sigma^x_{k+1}}/\tau_k)}$ $\forall j \in J^i_t$   \label{select}
            
            \STATE $\forall i : \sigma^i_{k+1} \leftarrow \sigma^{\mathrm{adopted}^i}_{k+1}, \pi^i_{k+1} \leftarrow \pi^{\mathrm{adopted}^i}_{k+1}$
            
            \STATE Take step $\forall i : a^i_t \sim \pi^i_{k+1}(\cdot|s^i_t), r^i_{t} = R(s^i_{t},a^i_{t},\hat{\mu}_t),s^i_{t+1} \sim P(\cdot|s^i_{t},a^i_{t},\hat{\mu}_t)$; $t \leftarrow t + 1$\label{stepInComm}
            
            \ENDFOR\label{end_comm}

    \ENDFOR
    \STATE \textbf{return} policies \{$\pi^i_K$\}$_{i=1}^{N}$
    \end{algorithmic}
\end{algorithm}

\subsection{Algorithm acceleration by use of experience-replay buffer}\label{buffer_section}

We modify our Alg. \ref{networked_algorithm} (and accordingly the algorithms of the two non-networked architectures) as follows, shown in \textcolor{blue}{\textit{blue}} in Alg. \ref{networked_algorithm_experience_replay}. Instead of using a transition $\zeta^i_{t-2}$ to compute the TD update within each $M_{pg}$ iteration and then discarding the transition, we store the transition in a buffer (Line \ref{store_transition}) until after the $M_{pg}$ loops. {Replay buffers are a common (MA)RL tool used especially with deep learning, precisely to improve data efficiency and reduce autocorrelation \citep{original_buffer,revisiting_replay,xu2024higher}.} 
When learning does take place in our modified algorithm (Lines \ref{beginlearningwithbuffer}-\ref{endlearningwithbuffer}), it involves cycling through the buffer for $L$ iterations - randomly shuffling the buffer between each - and thus conducting the TD update on each stored transition $L$ times. This allows us to {reduce the number of $M_{pg}$ loops}, as well as not requiring as small a learning rate $\{\beta_m\}_m$, allowing much faster learning in practice. Moreover, by shuffling the buffer before each cycle we reduce bias resulting from the dependency of samples along the continued, non-episodic system run, which may justify being able to {achieve adequate stable learning even when reducing the number of $M_{td}$ waiting steps} within each $M_{pg}$ loop (Sec. \ref{discussion}). 

Our replay buffer allows the first practical demonstrations of all three architectures for learning from a single continued system run - all of our experiments after those in Sec. \ref{ablation_buffer_experiment_section} use the buffer, with which we can see that learning can now occur. 
The intuition behind the better learning efficiency resulting from the buffer is as follows. The value of a state-action pair $p$ is dependent on the values of subsequent states reached, but the value of $p$ is only updated when the TD update is conducted on $p$, rather than every time a subsequent pair is updated. By learning from each stored transition multiple times, we not only make repeated use of the reward and transition information in each costly experience, but also repeatedly update each state-action pair in light of its likewise updated subsequent states. 

\begin{remark}
    Our introduction of the replay buffer means that the specific sample guarantees given in our theoretical results no longer apply. This is because these assume learning from a single stream of samples that are discarded straight after their first and only use, and these results do not account for temporary storage and shuffled reuse of samples. We do not directly update the sample guarantees in light of this algorithmic modification, and leave such additional proofs to future work. However we emphasise here that \textit{we expect the relative performances of the architectures to be preserved}, i.e. the central-agent architecture learns as fast or faster than the networked one, which in turn learns as fast or faster than the independent one. This is because, although the use of samples in learning has changed, the underlying machinery that drives the difference in performance between the architectures has not. The independent architecture will still have worse sample guarantees than the central-agent one due to bias caused by policy divergence, and our networked communication and adoption can reduce this divergence depending on network structure and the number of communication rounds, as before. In summary, \textit{the conceptual gap between our theoretical and empirical algorithms is minimal, and our theoretical results still give heuristic insight to explain our experimental results that use the buffer}, which show networked populations outperforming independent ones while underperforming or performing similarly to the central-agent populations, as predicted by our original theory. 
\end{remark}

We leave $\beta$ fixed across all iterations, as we found empirically that this yields sufficient learning. We have not experimented with decreasing $\beta$ as $l$ increases, though this may benefit learning. The transitions in the buffer are discarded after the replay cycles and a new buffer is initialised for the next iteration $k$, as in Line \ref{empty_buffer}. As such the space complexity of the buffer only grows linearly with the number of $M_{pg}$ iterations within each outer loop $k$, rather than with the number of $K$ loops.


\subsection{Generation of $\sigma^{i}_{k+1}$}\label{practical_generation}

Reducing the number of loops in the hope of achieving practical convergence times can lead to poorer estimation of the Q-function $\hat{Q}^i_{M_{pg}}$, and hence a greater variance in the quality of the updated policies $\pi^i_{k+1}$. This problem will increase with the size of the state and action spaces. In such cases we found empirically that an appropriate method for generating $\sigma^{i}_{k+1}$ dependent on $\pi^{i}_{k+1}$ allows our networked algorithm to significantly outperform the independent case by advantageously biasing the spread of particular policies. This is instead of generating $\sigma^{i}_{k+1}$ arbitrarily 
as required in the theoretical settings in Sec. \ref{properties}.

We do so via the steps added in \textcolor{orange}{\textit{orange}} in Alg. \ref{networked_algorithm_experience_replay}, which replace Line \ref{start_comm_theoretical} in Alg. \ref{networked_algorithm}: for $\boldsymbol\pi_{k+1}$ := ($\pi^1_{k+1},\dots,\pi^N_{k+1}$), we set $\sigma^{i}_{k+1}$ to a finite-step approximation $\widehat{\Psi}^i_{h,k+1}(\boldsymbol\pi_{k+1},\upsilon_0)$ of the discounted return $\Psi^i_{h,k+1}(\boldsymbol\pi_{k+1},\upsilon_0)$ (Def. \ref{psi_definition}). The approximation is given by, $\forall i,j \in \{1,\dots,N\}$ \[\widehat{\Psi}^i_{h,k+1}(\boldsymbol\pi_{k+1},\upsilon_0) =  \left[\sum^{E}_{e=0}\gamma^{e}(R(s^i_t,a^i_t,\hat{\mu}_t) \right.+h(\pi^i(s^i_t)))\left|\substack{t = t + e \\ a^j_t\sim \pi^j_{k+1}(s^j_t)\\ s^j_{t+1} \sim P(\cdot|s^j_{t},a^j_{t},\hat{\mu}_t)}\right].\] This is calculated by tracking each agent's discounted return for $E$ evaluation steps (Lines \ref{evaluation_begin}-\ref{evaluation_end}).

Generating $\sigma^{i}_{k+1}$ in this way means policies that are more likely to spread through the network are those estimated to receive a higher return in reality, despite being generated from poorly estimated Q-functions, biasing the population towards faster learning. 
Naturally the quality of the finite-step approximation depends on the number of evaluation steps $E$, but we found empirically that $E$ can be much smaller than $M_{pg}$ and still give marked convergence benefits.

\section{Experiments}\label{experiments}

Our technical contribution of the replay buffer to MFG algorithms for online learning from non-episodic system runs allows us also to {contribute the first empirical demonstrations of these algorithms, not just in the networked case but also in the central-agent and independent cases}. The latter two serve as baselines to show the advantages of the networked architecture. 
Experiments were conducted on a MacBook Pro, Apple M1 Max chip, 32 GB, 10 cores. We use \texttt{scipy.optimize.minimize} (employing Sequential Least Squares Programming) to conduct the optimisation step in Def. \ref{PMA_operator}, and the JAX framework to accelerate and vectorise some elements of our code. For reproducibility, our code is included in the publicly available Supplementary Material.

\subsection{Games}\label{objectives_appendix}

We follow the gold standard in prior works on stationary MFGs regarding the types of game demonstrated: we focus on grid-world environments where agents can move in the four cardinal directions or remain in place \citep{lauriere2021numerical,scalable_deep,zaman2023oraclefree,algumaei2023regularization,cui2023multiagent,wu2024populationaware}. While this type of experiment is characteristic of similar MFG works, we recognise that these are simple games. They nevertheless serve as useful preliminary demonstrations of the validity of our algorithms and the considerations necessary for achieving practical learning; we leave experiments in more complex environments to future work, which would likely require extending the algorithms to handle non-tabular Q-functions. Moreover, grid-world environments naturally reflect the deployed, spatial applications in which we are interested in our setting, where agents learn online and communicate with neighbours on a network (which is likely to be defined spatially, though is not restricted to such a case). 

We conduct numerical tests with two tasks (defined by the agents' reward functions), chosen for being particularly amenable to intuitive understanding of whether the agents are learning behaviours that are appropriate and explainable for the respective objective functions. In all cases, rewards are normalised in [0,1] after they are computed.

\paragraph{Cluster.} 

This is the inverse of the `exploration' game in \citep{scalable_deep}, where in our case agents are encouraged to gather together by the reward function $R(s^i_{t},a^i_{t},\hat{\mu}_t) =$ log$(\hat{\mu}_t(s^i_{t}))$. That is, agent $i$ receives a reward that is logarithmically proportional to the fraction of the population that is co-located with it at time $t$. We give the population no indication where they should cluster, agreeing this themselves over time.

\paragraph{Agree on a single target.} Unlike in the above `cluster' game, the agents are given options of locations at which to gather, and they must reach consensus among themselves. If the agents are co-located with one of a number of specified targets $\phi \in \Phi$ (in our experiments we place one target in each of the four corners of the grid), and other agents are also at that target, they get a reward proportional to the fraction of the population found there; otherwise they receive a penalty of -1. In other words, the agents must coordinate on which of a number of mutually beneficial points will be their single gathering place. The reward function is given by $R(s^i_{t},a^i_{t},\hat{\mu}_t) = r_{targ}(r_{collab}(\hat{\mu}_t(s^i_{t})))$, where \[r_{targ}(x)=\begin{cases}
          x \quad &\text{if} \, \exists \phi \in \Phi \text{ s.t. dist}(s^i_{t},\phi) = 0 \\
          -1 \quad &  \text{otherwise,} \\
     \end{cases} \]
\[r_{collab}(x)=\begin{cases}
          x \quad &\text{if} \, \hat{\mu}_t(s^i_{t}) > 1/N \\
          -1 \quad &  \text{otherwise.} \\
     \end{cases}
\]

These are both coordination games, where selfish agents can increase their individual rewards by following the same strategy as others and therefore have an incentive to communicate policies. Moreover, they require more sophisticated solutions than the dispersal/exploration games often considered in similar MFG works \citep{scalable_deep,zaman2023oraclefree,wu2024populationaware}, where a trivial starting policy that encourages agents to move across the grid at random may already be close to the equilibrium policy.

\subsection{Experimental metrics}\label{metrics_appendix}

To give as informative results as possible about both performance and proximity to the NE, we provide three metrics for each experiment. All metrics are plotted with 2-sigma confidence intervals (2 $\times$ standard deviation), computed over 10 trials (each with a random seed) of the system evolution in each setting. This is computed based on a call to \texttt{numpy.std} for each metric over each run.

\subsubsection{Exploitability}\label{exploitability}

Works on MFGs most commonly use the \textit{exploitability} metric to evaluate how close a given policy $\pi$ is to a NE policy $\pi^*$ \citep{continuous_fictitious_play,scaling_MFG,survey_learningMFGs,scalable_deep,algumaei2023regularization,wu2024populationaware}. The metric usually assumes that all agents are following the same policy $\pi$, and quantifies how much an agent could benefit by deviating from $\pi$, by measuring the difference between the return $V_h$ (Def. \ref{mean_field_return_def}) gained by $\pi$ and that gained by a policy that best responds to the population distribution generated by $\pi$. Let us denote by $\mu^{\pi}$ the distribution generated when $\pi$ is the policy followed by all of the population aside from the deviating agent; then the exploitability of policy $\pi$ is defined as follows:

\begin{definition}[Exploitability of $\pi$]\label{exploitability_definition}
The exploitability $\mathcal{E}$ of policy $\pi$ is given by:
\[\mathcal{E}(\pi) = \max_{\pi'}V_h(\pi',\mu^{\pi}) - V_h(\pi,\mu^{\pi}).\]\end{definition}
If $\pi$ has a large exploitability then an agent can significantly improve its return by deviating from $\pi$, meaning that $\pi$ is far from $\pi^*$, whereas an  exploitability of 0 implies that $\pi = \pi^*$ - i.e. lower exploitability is considered better.

Since we do not have access to the exact best response policy $\arg\max_{\pi'}V_h(\pi',\mu^{\pi})$ as in some related works \citep{scalable_deep,wu2024populationaware}, we instead approximate the exploitability metric, similarly to \citep{flock}, as follows. We freeze the policy of all agents apart from a deviating agent, for which we store its current policy and then conduct 40 `deviation' $k$ loops of policy improvement. To approximate the expectations in Def. \ref{exploitability_definition}, we take the best return of the deviating agent across the 40 $k$ loops, as well as the mean of all the other agents' returns across these same loops. We then revert the agent back to its stored policy, before learning continues for all agents. Due to the expensive computations required for this metric, we evaluate it only on alternate $k$ iterations of the actual system evolution (for our ablation study of the experience replay buffer in Sec. \ref{ablation_buffer_experiment_section}, we evaluate only every 20 $k$). 

Since prior works conducting empirical testing have generally focused on the centralised setting, evaluations have not had to consider the exploitability metric when not all agents are following a single policy $\pi_k$, as may occur in the independent or networked settings, i.e. when $\pi^i_k \neq \pi^j_k$ for $i,k \in \{1,\dots,N\}$. The method described above for approximating exploitability involves calculating the mean return of all non-deviating agents' policies. While this is $\pi_k$ in the centralised case, if the non-deviating agents do not share a single policy, then this method is in fact approximating the exploitability of their joint policy $\boldsymbol\pi^{-d}_k$, where $d$ is the deviating agent. 

\textit{The exploitability metric has a number of limitations in our setting}. In coordination games (the setting for our tasks), agents benefit by following the same behaviour as others, and so a deviating agent generally stands to gain less from a `best-responding' policy than it might in the non-coordination games on which many other works focus. For example, the return of a best-responding agent in the `cluster’ game still depends on the extent to which other agents coordinate on where to cluster, meaning it cannot significantly increase its return by deviating from a badly clustering policy. This means that the downward trajectory of the exploitability metric is less clear in our plots than in other works. 

Moreover, our approximation takes place via policy improvement steps (as in the main algorithm) for an independent, deviating agent while the policies of the rest of the population are frozen. As such, the quality of our approximation is limited by the number of policy-improvement/expectation-estimation rounds, which must be restricted for the sake of the running speed of the experiments. Moreover, since one of the findings of our paper is that networked agents can improve their policies faster than independent agents, it is arguably unsurprising that approximating the best response by an independently deviating agent sometimes gives an unclear and noisy metric.

Given the limitations presented by approximating exploitability, we also provide the second metric to indicate the progress of learning.

\subsubsection{Average discounted return}

We record the average finite-step discounted return of the agents' policies $\pi_k^i$ during the $M_{pg}$ steps of each outer $k$ loop. This allows us to observe that settings that converge to similar exploitability values may not have similar average agent returns, suggesting that some algorithms are better than others not just at finding equilibria, but also at finding preferable equilibria (when the assumption of a unique MFG-NE is removed by reducing regularisation; see Sec. \ref{discussion}) - cf. \citet{graber2025trembling,li2025repositioning}. See, for example, Fig. \ref{cluster_16}, where the networked agents converge to similar exploitability as the independent agents, but receive higher average reward.

\subsubsection{Policy divergence}

We record the population's average policy divergence $\frac{1}{N}\Delta_k$ := $\frac{1}{N}\sum_{i=1}^N||\pi_k^i - {\pi_k^1}||_1$ for the arbitrary policy $\bar{\pi} = {\pi^1}$. Many of our theoretical results and proofs relate to the policy divergence, and in Sec. \ref{properties} we show extensively how the comparatively worsening sample complexities between the centralised, networked and independent cases are the result of their range of policy divergences. We therefore include this metric to show how this relationship affects learning in practice.

Furthermore, the theoretical guarantees assume that the population is trying to learn the unique equilibrium policy $\pi^*$, with the implication that all agents should end up with this identical policy, regardless of the learning architecture (Sec. \ref{properties}). However, we find in practice that populations may be converging (in terms of exploitability/return) while having non-diminishing policy divergence, particularly in the independent setting. We therefore also include this metric to indicate the difference between theoretical and empirical convergence.

\subsection{Hyperparameters}\label{lambda}
\begin{table*}
\vspace{-0.5cm}
\caption{Hyperparameters}\label{Hyperparameters}
    \centering
    \begin{tabular}{p{0.8cm}p{0.8cm}p{13.8cm}}\\ 
\toprule
  Hyper-param.     & Value     & Comment \\
  \midrule
Grid-size     & 8x8 / 16x16 & Most experiments are run on the smaller grid, while Figs. {\ref{cluster_16} and \ref{agree_16}} demonstrate learning in a larger state space.    \\
    
    \midrule
    Trials & 10 & We run 10 trials with different random seeds for each experiment. We plot the mean and 2-sigma error bars for each metric across the trials.   \\
    \midrule
    Pop.     & 250 & We tested $N$ in \{25,50,100,200,250\}, with the networked architecture generally performing equally well with all population sizes $\geq$ 50. We chose 250 for our demonstrations, to show that our algorithm can handle large populations, indeed often larger than those demonstrated in other mean-field works, especially for grid-world environments  \citep{mfrl_yang,rl_stationary,subramanian2022multi,pomfrl,approximately_entropy,subjective_equilibria,DMFG,general_framework,cui2023multiagent}. 
    In experiments testing robustness to population increase, the population instead begins at 50 agents and has 200 added at the marked point.
    \\ 
    \midrule
    $K$     & 200 / 400       & $K$ is chosen to be large enough to see exploitability reducing, and converging where possible.\\
    \midrule
    $M_{pg}$     & 500 / 1000       &  We wish to illustrate the benefits of our networked architecture and replay buffer in reducing the number of loops required for convergence, i.e. we wish to select a low value that still permits learning. We tested $M_{pg}$ in \{300,500, 600,800,1000,1200,1300,1400,1500,1800,2000,2500,3000\}, and chose 500 for demonstrations on the 8x8 grids, and 1000 for the 16x16 grids. It may be possible to optimise these values further in combination with other hyperparameters. \\ 
    \midrule

    $M_{td}$     &  1  & We tested  $M_{td}$  in \{1,2,10,100\}, and found that we could still achieve convergence with $M_{td} = 1$. This is much lower than the requirements of the theoretical algorithms, essentially allowing us to remove the innermost nested learning loop.\\
     
    \midrule
    $C$     &   1   & We tested $C$ in \{1,5,10\}. We choose 1 to show the convergence benefits brought by even a single communication round, even in networks that may have limited connectivity.\\ 
\midrule
     $L$     &   100    &  As with $M_{pg}$, we wish to select a low value that still permits learning. We tested $L$ in \{50,100,200,300,400,500\}. In combination with our other hyperparameters, we found $L \leq 50$ led to less good results, but it may be possible to optimise this hyperparameter further.\\
    \midrule
    $E$     &   100    & We tested $E$ in \{100,300,1000\}, and choose the lowest value to show the benefit to convergence even from few evaluation steps. It may be possible to reduce this value further and still achieve similar results. \\
     \midrule
    $\gamma$     & 0.9 & Standard choice across RL literature.     \\
    \midrule
    
    $\beta$     &   0.1    & We tested $\beta$ in \{0.01,0.1\} and found 0.1 to be small enough for adequate learning at an acceptable speed.  Further optimising this hyperparameter (including by having it decay with increasing $l \in 0,\dots,L-1$, rather than leaving it fixed) may lead to better results.\\ 
    \midrule
    $\eta$     &   0.01     & We tested $\eta$ in \{0.001,0.01,0.1,1,10\} and found that 0.01 gave stable learning that progressed sufficiently quickly. 
    \\
    \midrule
    $\lambda$    &    0    & We tested $\lambda$ in \{0,0.0001,0.001,0.01,0.1,1\}. {Since we can reduce $\lambda$  to 0 with no detriment to empirical convergence, we do so in order not to bias the NE.}\\
    \midrule
    $\tau_k$    &   cf. comment    & For fixed $\tau_k$ $\forall k$, we tested \{1,10,100,1000\}. In our experiments for fixed $\tau_k$ the value is 100 (see Figs. \ref{cluster_fixed} and \ref{agree_fixed}); this yields learning, but does not perform as well as if we anneal $\tau_k$ as follows. We begin with $\tau_0 = 10000 / (10 ** \lceil(K-1)/10\rceil$), and multiply $\tau_k$ by 10 whenever $k\bmod10=1$ i.e. every 10 iterations. Further optimising the annealing process may lead to better results.
    \\
    \bottomrule
    \end{tabular}
    
    \end{table*}

See Table \ref{Hyperparameters} for our hyperparameter choices. In general, we seek to show that our networked algorithm is robust to `poor' choices of hyperparameters, such as low numbers of iterations, as may be required when aiming for practical convergence times in complex real-world problems. By contrast, the convergence speed of the independent algorithm suffers much more significantly without idealised hyperparameter choices. As such, our experimental demonstrations in the plots generally involve hyperparameter choices at the low end of the values we tested during our research.

\subsection{Results and discussion}\label{discussion}

We first provide results for learning without a replay buffer as in the theoretical algorithms. We then give the rest of our experiments with the replay buffer, beginning with our standard experimental setting and then robustness studies and numerous ablations (note that the independent setting serves as an implicit ablation of our communication scheme). We summarise findings in the body of each sub-section, while the specific results are discussed fully in each figure's caption. In each plot the decimals refer to each agent's broadcast radius as a fraction of the maximum possible distance in the grid (i.e. the diagonal). Note that the networked population with the largest radius is always fully connected.

We pre-empt possible concerns regarding the wide confidence intervals in many of our plots by saying that many works with similar experiments do not report error bars at all, and if they do they usually only give 1-sigma intervals, whereas we give  2-sigma  \citep{lauriere2021numerical,scalable_deep,zaman2023oraclefree,algumaei2023regularization,cui2023multiagent}. Moreover, the central-agent architecture usually has similar or higher variance to the networked agents in the plots, indicating that this is not an issue introduced by our communication algorithm; it is instead likely to be due to poor estimation of the Q-function when using the small numbers of loops required for practical runtimes. The independent agents have very low variance, but this is because they hardly appear to increase their returns at all in most cases.

We also give the following remark regarding the exploitability metric in some of our experimental plots, relating to the issues with this metric in coordination games, as discussed in Sec. \ref{exploitability}:

\begin{remark}\label{increasing_exploitability_remark}
The reward structure of our coordination games is such that exploitability sometimes increases from its initial value before it decreases down to 0 (e.g. Fig. \ref{cluster_regular}). This is because agents are rewarded proportionally to how many other agents are co-located with them: when agents are evenly dispersed at the beginning of the run, it is difficult for even a deviating, best-responding agent to significantly increase its reward. However, once some agents start to aggregate, a best-responding agent can take advantage of this to substantially increase its reward (giving higher exploitability), before all the other agents catch up and aggregate at a single point, reducing the exploitability down to 0. Due to this arc, in some of our plots the independent case may have lower exploitability at certain points than the other architectures, but this is not necessarily a sign of good performance. In fact, in such cases we can often see that the independent agents are hardly learning at all, with the independent agents' average return not increasing and the exploitability staying level rather than ultimately decreasing (see, for example, Figs. \ref{cluster_regular}, \ref{cluster_fail}, \ref{cluster_add} and \ref{cluster_16}). 
\end{remark}

\subsubsection{Learning with no experience replay buffer}\label{ablation_buffer_experiment_section}

\begin{figure}[t]
  \centering
  \includegraphics[width=1.0\linewidth]{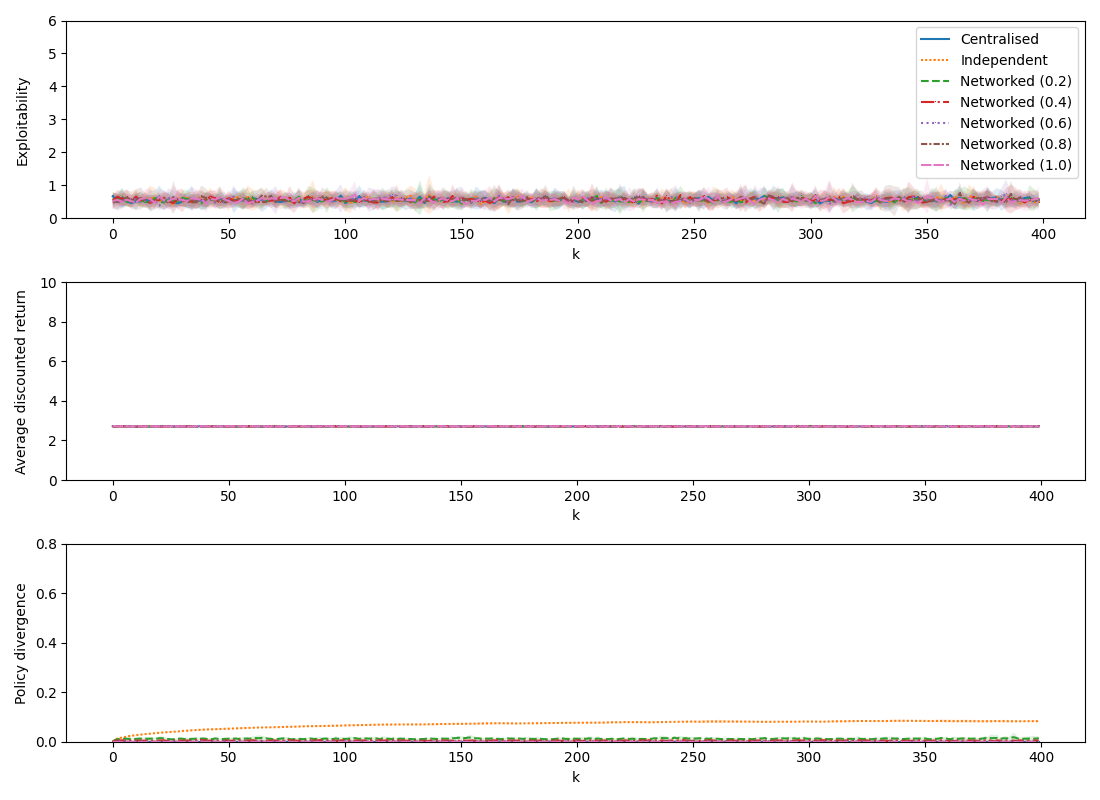}\caption{`Cluster' game without our experience replay buffer. There is no noticeable improvement in any of the agents' returns, i.e. no noticeable learning, even after $K = 400$ iterations.}
  \label{cluster_no_buffer}
\end{figure}

\begin{figure}[t]
  \centering
  \includegraphics[width=1.0\linewidth]{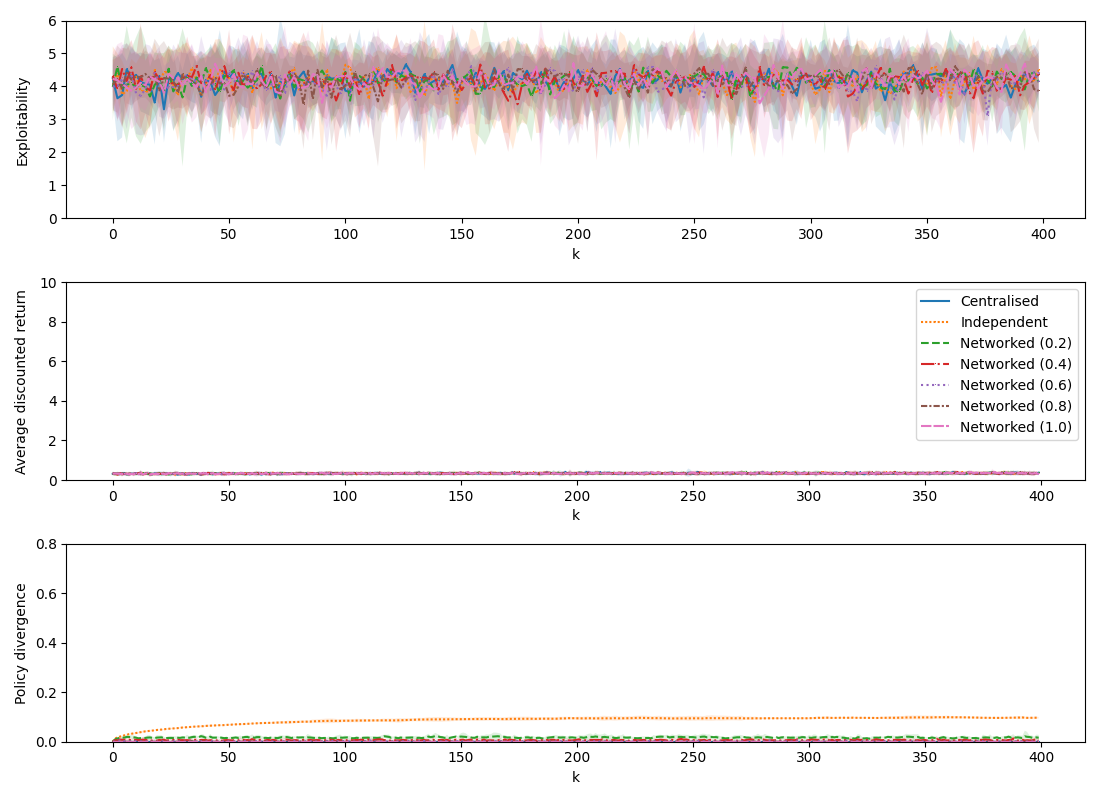}\caption{`Target agreement' game without our experience replay buffer. There is no noticeable improvement in any of the agents' returns, i.e. no noticeable learning, even after $K = 400$ iterations.
  }
  \label{agree_no_buffer}
\end{figure}

Figs. \ref{cluster_no_buffer} and \ref{agree_no_buffer} illustrate the importance of our incorporation of the experience replay buffer. Without it, as in the original theoretical version of the algorithms, there is no noticeable improvement in any of the agents' returns, i.e. no noticeable learning, even after $K = 400$ iterations. In these experiments without a replay buffer we run the core learning section of the algorithm as in Lines \ref{begin_core_theoretical}-\ref{end_core_theoretical} of Alg. \ref{networked_algorithm}, keeping the hyperparameters the same as in our main experiments, i.e. $M_{pg} = 500$, $M_{td}=1$, etc. (see Tab. \ref{Hyperparameters}). The theoretical results in fact require that these values are many orders of magnitude higher. While setting them as such might mean that empirically we do see some learning occurring within $K<400$, such experiments would take impractically long to run on standard computers, taking many days or even many weeks.

These experiments are run for 5 trials rather than 10 as in all other cases, and with exploitability evaluated every 20 $k$ instead of every 2 $k$ for computational efficiency. 

The remainder of our experiments all include our replay buffer, and therefore do permit learning in practical time, albeit at different rates for the different architectures.

\begin{figure}[t]
  \centering
  \includegraphics[width=1.0\linewidth]{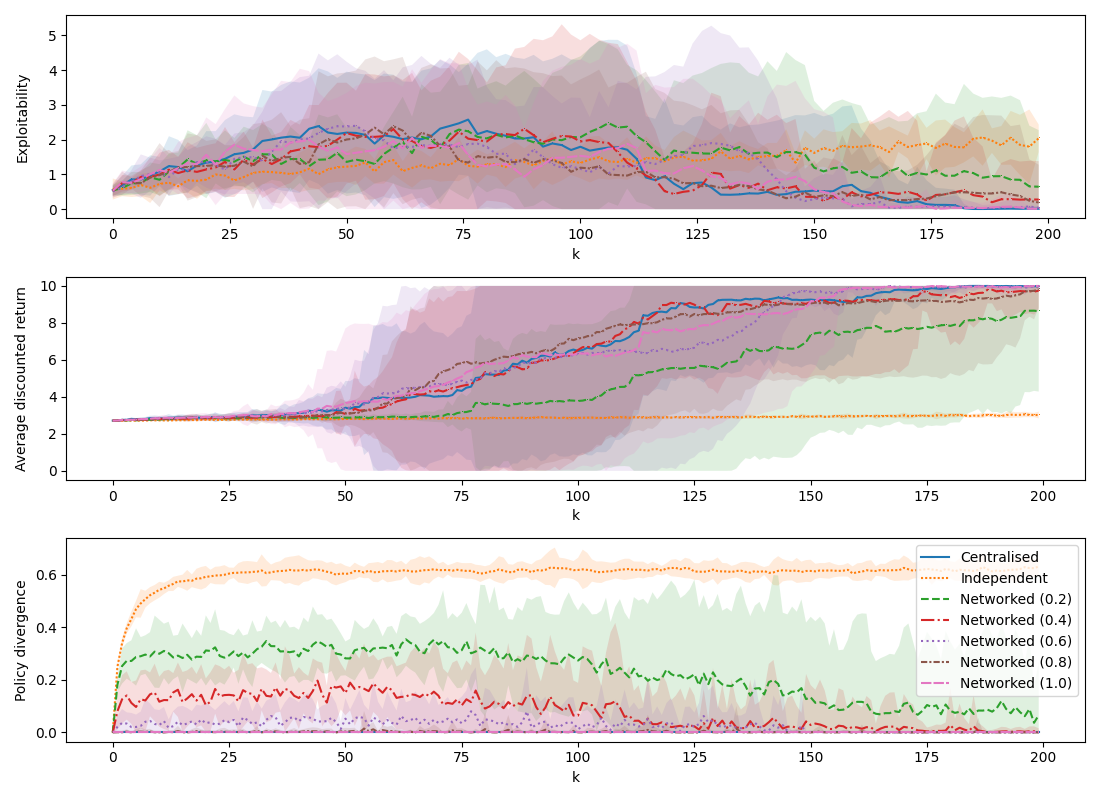}
  \caption{`Cluster' game. {Even with only a single communication round, our networked architecture significantly outperforms the independent case}, which hardly appears to be learning at all. All broadcast radii except the smallest (0.2, green) have similar mean exploitability and return to the centralised case.}
  \label{cluster_regular}
\end{figure}

\begin{figure}[t]
  \centering
  \includegraphics[width=1.0\linewidth]{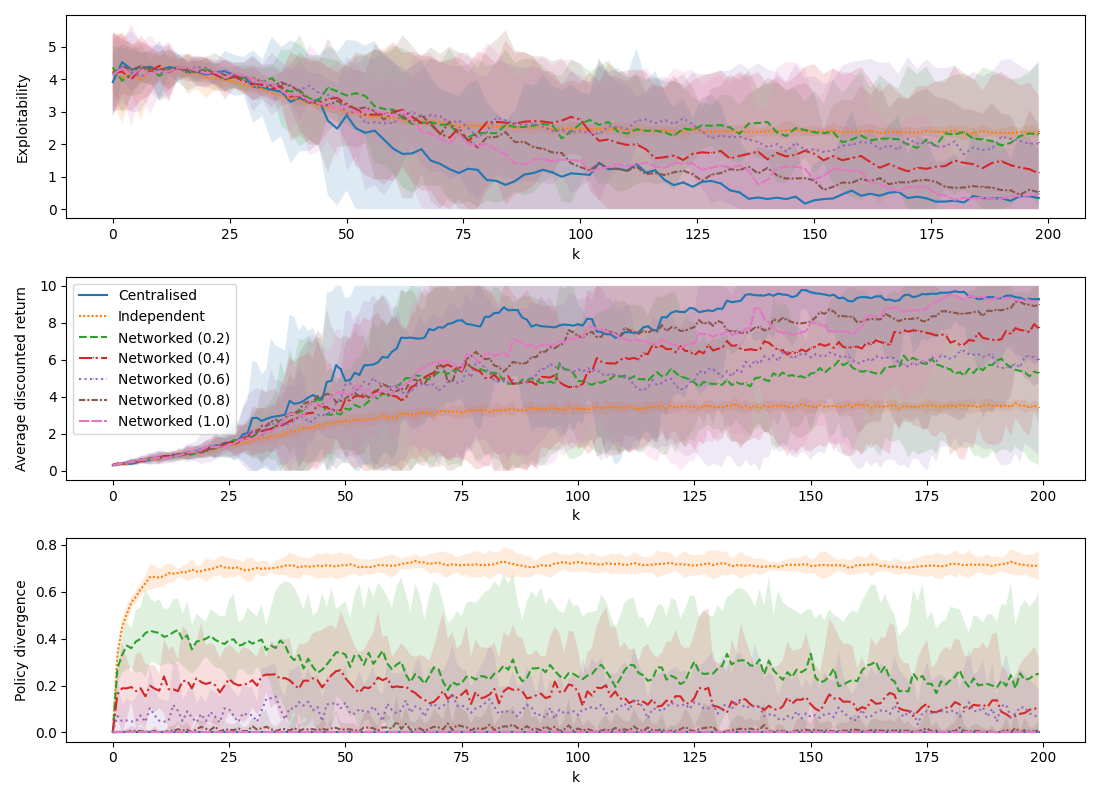}\caption{`Target agreement' game. {Even with only a single communication round, our networked case outperforms the independent case} with respect to exploitability and return. The fact that the lowest broadcast radius (0.2, green) ends with similar exploitability to the independent case yet higher return suggests {our networked algorithm might help agents find `preferable’ equilibria}. 
  }
  \label{agree_regular}
\end{figure}

\subsubsection{Standard experimental setting with replay buffer}

Even with only a single communication round in each of the $K$ loops, networked agents learn faster and reach higher returns than independent agents, which hardly appear to learn at all. Moreover networked agents appear to match the central-agent population in the `cluster' game (Fig. \ref{cluster_regular}). Our experiments show that our practical algorithmic enhancements enable convergence within a practical number of iterations even when we remove a number of the assumptions required for the theoretical algorithms:

\begin{itemize}
     \item We reduce $M_{pg}$ by many orders of magnitude from its theoretically required value (see Sec. \ref{practical_enhancements}), while still converging within a reasonable $K$. We keep the learning rate $\beta$ fixed, removing the annealing scheme for $\{\beta_m\}_{m \in \{0,\dots,M_{pg}-1\}}$ required in the theorems, and use a much higher value.

     \item In our experiments we do not ensure that the communication network $\mathcal{G}_t$ remains static and connected, nor that the diameter $d_\mathcal{G}$ of the network is equal for all $k$. Nevertheless, even with a single communication round the networked agents learn faster than independent ones (which hardly learn at all), sometimes performing similarly to the centralised case.

     \item The $M_{td}$ parameter is theoretically required for the learner to wait between collecting samples when learning from the empirical distribution in a non-episodic system run. However, our replay buffer allows us to reduce it to 1, effectively {removing the innermost loop of the nested learning algorithm} (see Line \ref{mtd_loop_theoretical} of Alg. \ref{networked_algorithm}). 
    
    \item We can reduce the scaling parameter $\lambda$ of the entropy regulariser to 0, i.e. we converge even without regularisation, allowing us to leave the MFG-NE unbiased and also removing Assumption \ref{Persistence_excitation}. 
    {In general an unregularised MFG-NE is not unique \citep{policy_mirror_independent}; the ability of centralised and networked agents to coordinate on one of the multiple possible solutions may explain why they outperform the independent case, as discussed further below (cf. \citet{graber2025trembling,li2025repositioning}).}
    \item For the PMA operator (Def. \ref{PMA_operator}), we conduct the optimisation over the set $u \in \Delta_{\mathcal{A}}$ instead of $u \in \mathcal{U}_{L_h}$, i.e. we can {choose from all possible distributions over actions instead of needing to identify the Lipschitz constants given in Assumption \ref{lip_cont_assumpt}}.
\end{itemize}

We now give further intuition into the benefits of our communication scheme in our empirical settings where multiple equilibria are possible. For sufficiently high $\lambda$ the MFG-NE is unique, and involves all the agents constantly moving about with high entropy, at the cost of biasing the problem. However, when $\lambda$ is 0, the `target agreement' and `cluster' tasks both explicitly admit multiple Nash equilibria. In a given trial of the `target agreement' task, all the agents could converge to remaining stationary at any one of the four corners, and any one of these four situations would lead to the highest possible returns. We found in our experiments that with the different random seeds for each trial, agents did end up converging to a different corner at random each time. Similarly in the `cluster' task: for a given trial all the agents could converge to remaining stationary in any one of the grid points, and any one of these $height\times width$ situations would lead to the highest possible returns. (In practice, empirically we found that the agents usually converged at random to one of the corners in the `cluster' task as well, rather than to anywhere on the grid. This is because in the early stages of the trial, when agents start with random policies, they already spend more time visiting corners, because at any corner three actions will keep them in place, since they cannot move off the edge of the grid).

The discussion so far applies to \textit{Nash} equilibria, i.e. the situations where agents end up with the highest possible returns (equivalent to a normalised average return of 10 in the plots). Population distributions can also be at an equilibrium that is not Nash nor one that receives particularly high returns: we can broadly characterise three situations here:

\begin{enumerate}
    \item Agents, which begin the trial with random policies, never manage to reach any critical mass that breaks the ties between the possible coordination points, so continue moving about the grid with a high degree of entropy forever, even if $\lambda$ is 0. This is most likely what is happening for the independent agents across the experiments, and is why they usually converge to low returns. 
    \item The population gets segregated into two or more isolated parts of the grid, each of which would otherwise give a Nash equilibrium if the whole population were present e.g. half the population learns a policy that remains in the top left corner while the other half learns to stay at the bottom right. If the policies do not retain enough exploration, the agents will never discover the other isolated groups with which they could combine for mutual benefit (whilst if there is too much exploration, we revert to one of the other suboptimal situations, depending on the value of $\lambda$). 
    \item The population is not segregated, but oscillates between two or more locations that would otherwise represent Nash equilibria, without ever being able to settle on stable policies that agree on one location. This is similar to Case 1, but with the number of meeting points that are visited having been narrowed down. 
\end{enumerate}

Case 1 is likely to receive the worst returns. How much worse Case 2 and 3 are than the Nash equilibria depends on the size of the segregated populations and/or the frequency of the visitations caused by the oscillations. The ability of learning architectures to align the behaviour of the \textit{whole} population on a \textit{single} choice of Nash equilibrium location determines how close to the maximum return the population will receive. The independent case has no way to align policies outside of the signal from the returns themselves; if no critical mass ever forms to show differentiation in the returns, then the independent population will always remain at a low performing equilibrium. The central-agent case has an inherent method for aligning the policies of the whole population, but these policies may still oscillate between locations that would otherwise be Nash equilibria, which is why central-agent populations do not always reach the maximum returns in our plots. 

Our communication algorithm provides a method both for 1) aligning agents' policies, and for 2) choosing better performing policies on which to align (where both of these elements contribute to the selection of better equilibria). This is why we see our decentralised, networked populations receiving higher returns than the independent ones, as our algorithm helps agents to get out of the worse performing equilibria. (In principle, under the right conditions, our communication paradigm could even outperform the central-agent case: the latter aligns the population on a policy update of arbitrary quality, generated by arbitrary agent $i=1$, rather than aligning on better performing policies.) The degree to which our communication algorithm leads to policy consensus depends upon the network connectedness and the number of communication rounds. Since in our experiments we use $C=1$, it is the network connectedness - determined by the size of the broadcast radius - that has the greatest effect (for greater numbers of communication rounds, this may matter less). This is why we see the populations with higher broadcast radii converging to higher returns faster than populations with lower broadcast radii, which are in turn more capable than entirely independent agents - they are better able to align the population so as to converge to equilibria that are closer to optimal (Nash) equilibria. 

In summary, the fact that different populations in our experiments do not just improve their returns at different speeds, but actually appear to converge to different final returns, is reflective of them settling at different equilibria that give different returns. Our communication algorithm actively helps populations to settle at equilibria that are closer to optimal, i.e. `preferable' (and in so doing, to choose between multiple possible Nash equilibria).



\subsubsection{Robustness experiments}\label{robustness_experiments}

We consider two scenarios to which we desire real-world many-agent systems (e.g. robotic swarms, autonomous vehicle traffic, etc.) to be robust
. The networked setup affords population \textbf{fault-tolerance} and \textbf{online scalability}, which are motivating qualities of many-agent systems.

\paragraph{Fault-tolerance} We consider a scenario in which the learning/updating procedure of agents fails with a certain probability within each iteration, in which cases $\pi^{i}_{k+1} = \pi^{i}_{k}$ (see Figs. \ref{cluster_fail} and \ref{agree_fail} for our experimental results in this scenario). In real-life decentralised settings, this might be particularly liable to occur since the updating process might only be synchronised between agents by internal clock ticks, such that some agents may not complete their update in the allotted time but will nevertheless be required to take the next step in the environment. Regardless of their cause, such failures slow the improvement of the population in the independent case, and in the central-agent population it means no improvement occurs at all in any iteration in which failure occurs, as there is a single point of failure.  Networked communication instead provides redundancy in case of update failures, with the updated policies of any agents that have managed to learn spreading through the population to those that have not (cf. \citet{10971233}). This feature thus ensures that improvement can continue for potentially the whole population even if a high number of agents do not manage to learn at a given iteration. 

Our experimental setup for this scenario is as follows: at every $k$ iteration each learner (whether centralised or decentralised) fails to update its policy (i.e. Line \ref{end_core_theoretical} of Alg. \ref{networked_algorithm} is not executed such that $\pi^{i}_{k+1} = \pi^{i}_{k}$) with a 50\% probability. The communication network allows agents that have successfully updated their policies to spread this information to those that have not, providing redundancy that the centralised and independent settings do not have. See Figs. \ref{cluster_fail} and \ref{agree_fail}.

\paragraph{Online scalability} We may want to arbitrarily increase the size of a population of agents that are already learning or operating in the environment (we can imagine extra fleets of autonomous cars or drones being deployed) - see Sec. \ref{related_work} for comparison with other works considering this type of robustness \citep{dawood2023safe,10.1002/aaai.12131,10.1007/978-981-97-1087-4_9,wu2024populationaware}. A purely independent setting would require all the new agents to learn a policy individually given the existing distribution, and the process of their following and improving policies from scratch may itself disturb the MFG-NE that has already been achieved by the original population. 
With a communication network, however, the policies that have been learnt so far can quickly be shared with the new agents in a decentralised way, hopefully before their unoptimised policies can destabilise the current MFG-NE. This would provide, for example, a way to bootstrap a large population from a smaller pre-trained group, if training were considered expensive in a given setting. 

Our experimental setup for this scenario is as follows: instead of having 250 agents throughout, the population begins with 50 agents learning normally, and a further 200 agents are added to the population at the marked point. The networked architectures are quickly able to spread the learnt policies to the newly arrived agents such that learning progress is minimally disturbed, whereas convergence is significantly impacted in the independent case. See Figs.  \ref{cluster_add} and \ref{agree_add}.

\begin{figure}[t]
  \centering
    \includegraphics[width=1.0\linewidth]{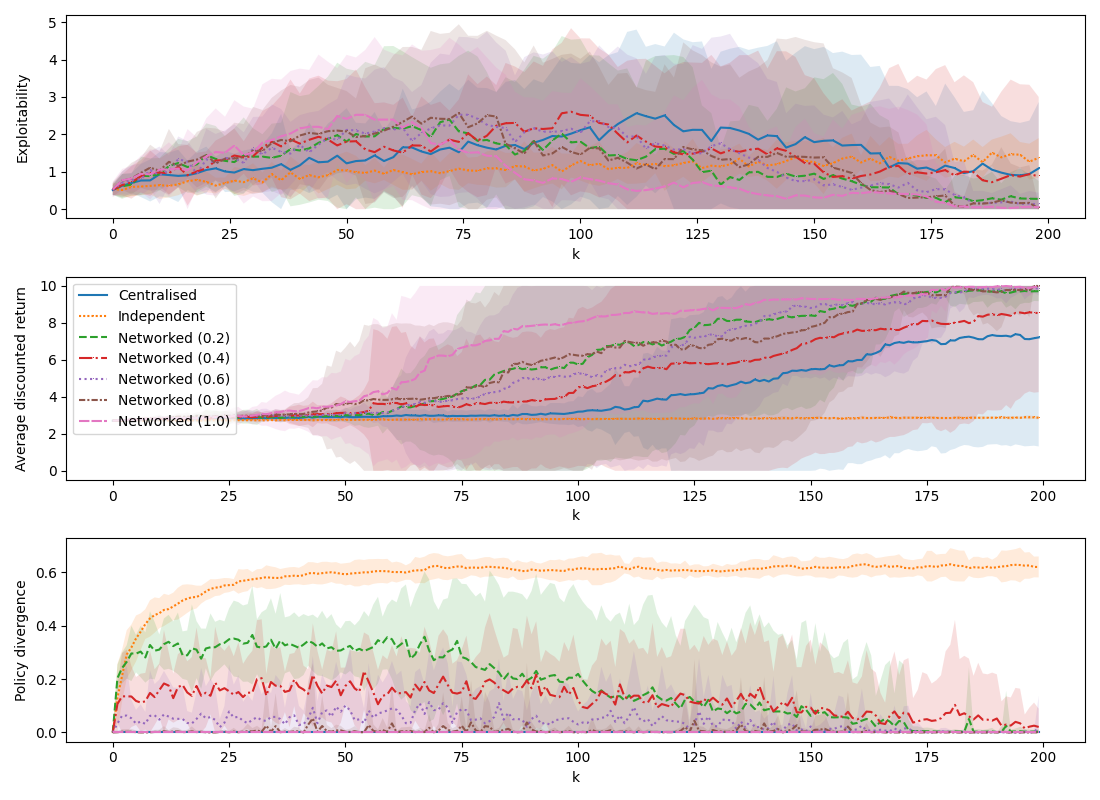}
  \caption{`Cluster' game, testing robustness to 50\% probability of policy update failure. The communication network allows agents that have successfully updated their policies to spread this information to those that have not, providing redundancy. Independent learners cannot do this and hardly appear to learn at all (no increase in return); likewise the centralised population is susceptible to its single point of failure and learns slower than before. Thus {our networked architecture outperforms both the centralised and independent cases}.}
  \label{cluster_fail}
 \end{figure}

\begin{figure}[t]
  \centering 
  \includegraphics[width=1.0\linewidth]{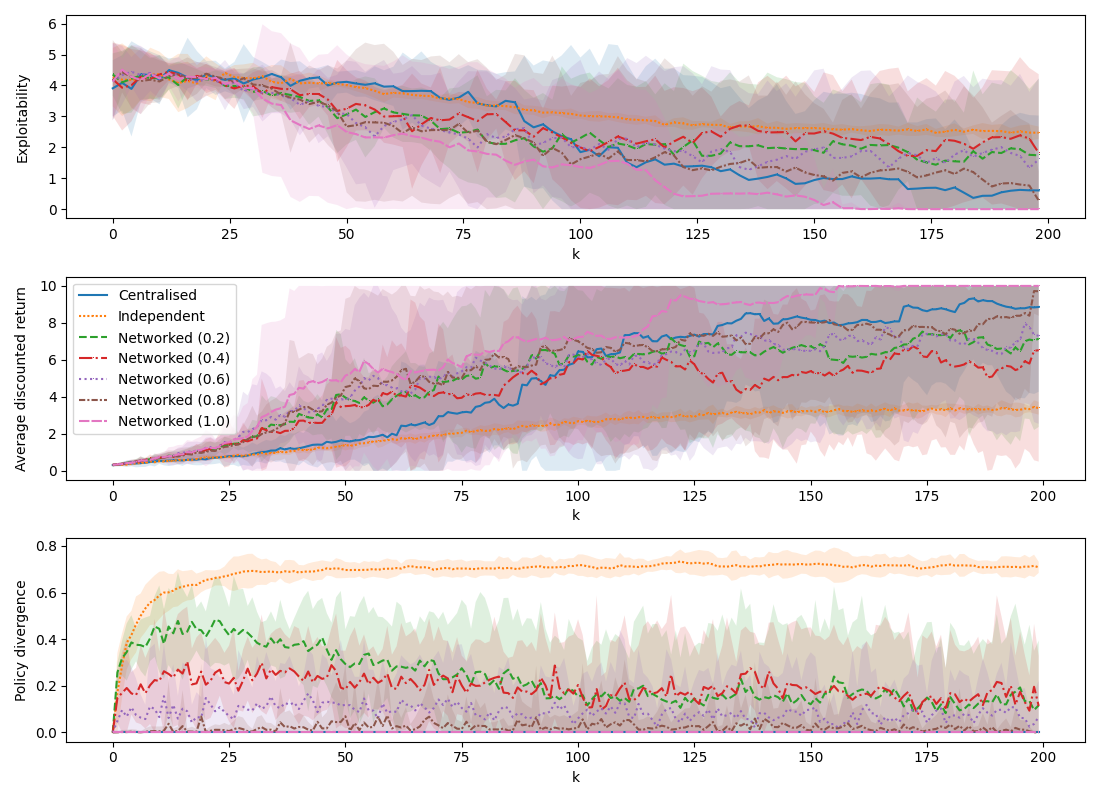} 
    \caption{`Target agreement' game, testing robustness to 50\% probability of policy update failure. {All the networked cases outperform the independent case and also learn faster than the centralised case for long periods}. The communication network allows agents that have successfully updated their policies to spread this information to those that have not, providing redundancy. Independent learners cannot do this so have even slower convergence than normal in this task; likewise the centralised architecture is susceptible to its single point of failure, hence learning can be slower than in the networked case.
}
    \label{agree_fail}
  
\end{figure}

\begin{figure}[t]
  \centering
  \includegraphics[width=1.0\linewidth]{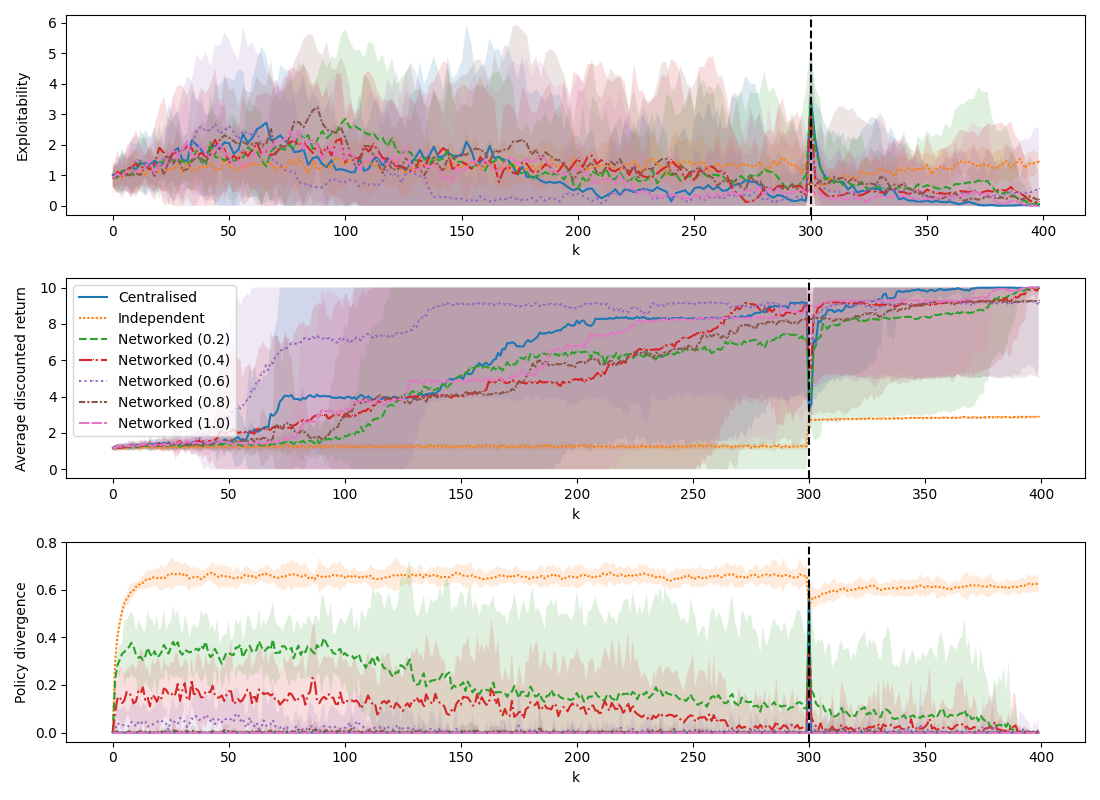}
  \caption{`Cluster' game, testing robustness to a five-times increase in population. While the independent algorithm appears to enjoy similar exploitability to the other cases (see Rem. \ref{increasing_exploitability_remark}), we can see from its average return that it is not in fact learning at all; while the return rises after the increase in population size this is only because there are now more agents with which to be co-located, rather than because learning has progressed. Since here, unlike in the `target agreement’ game in Fig. \ref{agree_add}, independent agents have hardly improved their return in the first place, we do not see the adverse effect that the addition of agents to the population has on the progress of learning. All networked populations {perform similarly to or outperform the centralised case, and all markedly outperform the independent case} in terms of return. The communication network allows the learnt policies to quickly spread to the newly arrived agents, such that the progression of learning is minimally disturbed, without needing to rely on the assumption of a centralised learner. The fact that, in all cases, the return prior to the population increase at $k=300$ is lower than in Fig. \ref{cluster_regular}, is reflective of the fact that the error in the solution reduces as $N$ tends to infinity.
  }
  \label{cluster_add}
\end{figure}

\begin{figure}[t]
    \centering
    \includegraphics[width=1.0\linewidth]{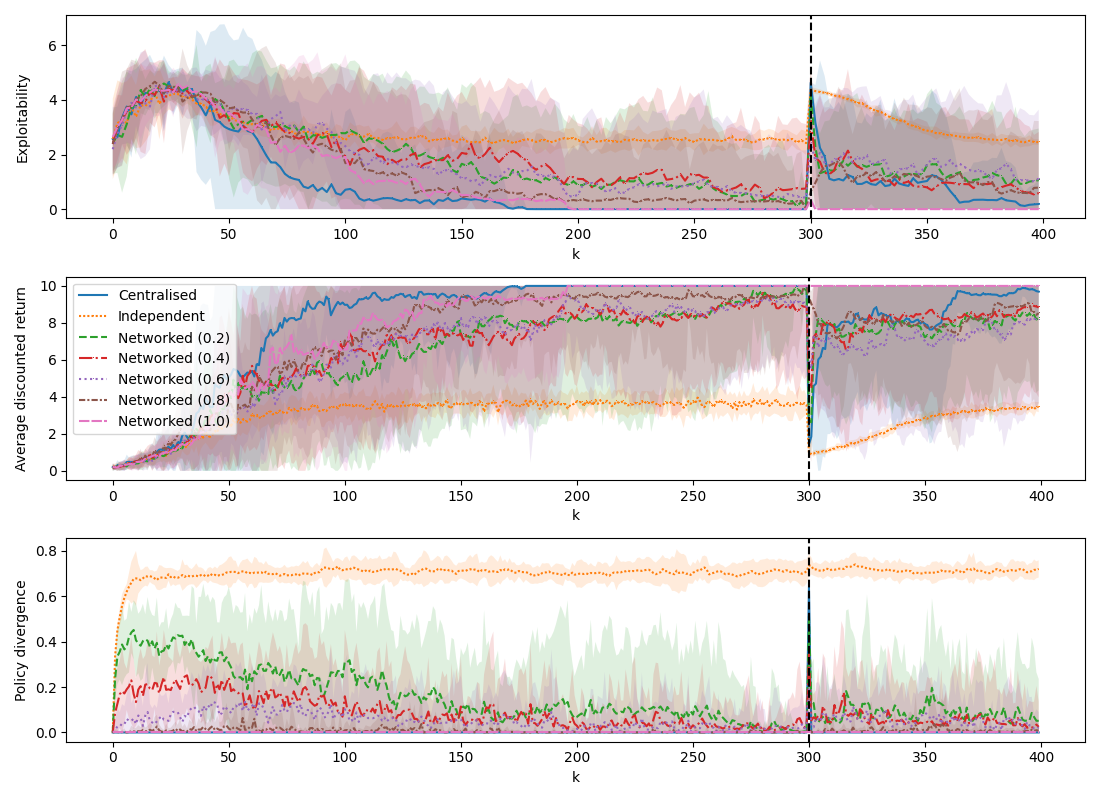} 
    \caption{`Target agreement' game, testing robustness to a five-times increase in population. The networked architectures are quickly able to spread the learnt policies to the newly arrived agents such that {learning progress is minimally disturbed, whereas convergence is significantly impacted in the independent case. The largest broadcast radius (1.0, pink), in particular, suffers no disturbance at all, being more robust than the centralised case}, which takes a significant amount of time to return to equilibrium.} 
    \label{agree_add}
\end{figure}

The remainder of our experiments provide further studies and ablations in the standard settings (i.e. not the robustness scenarios):

\subsubsection{Experiments on larger grid}

Figs. \ref{cluster_16} and \ref{agree_16} show the result of learning on a grid of size 16x16 instead of 8x8 as in all other experiments. There is at times greater differentiation in this setting than in the 8x8 grid between the performances of the different broadcast radii of the networked architecture (as is to be expected in a less densely populated environment). The {networked architecture continues to outperform the independent case for most broadcast radii}.

\begin{figure}[t]
  \centering
  \includegraphics[width=1.0\linewidth]{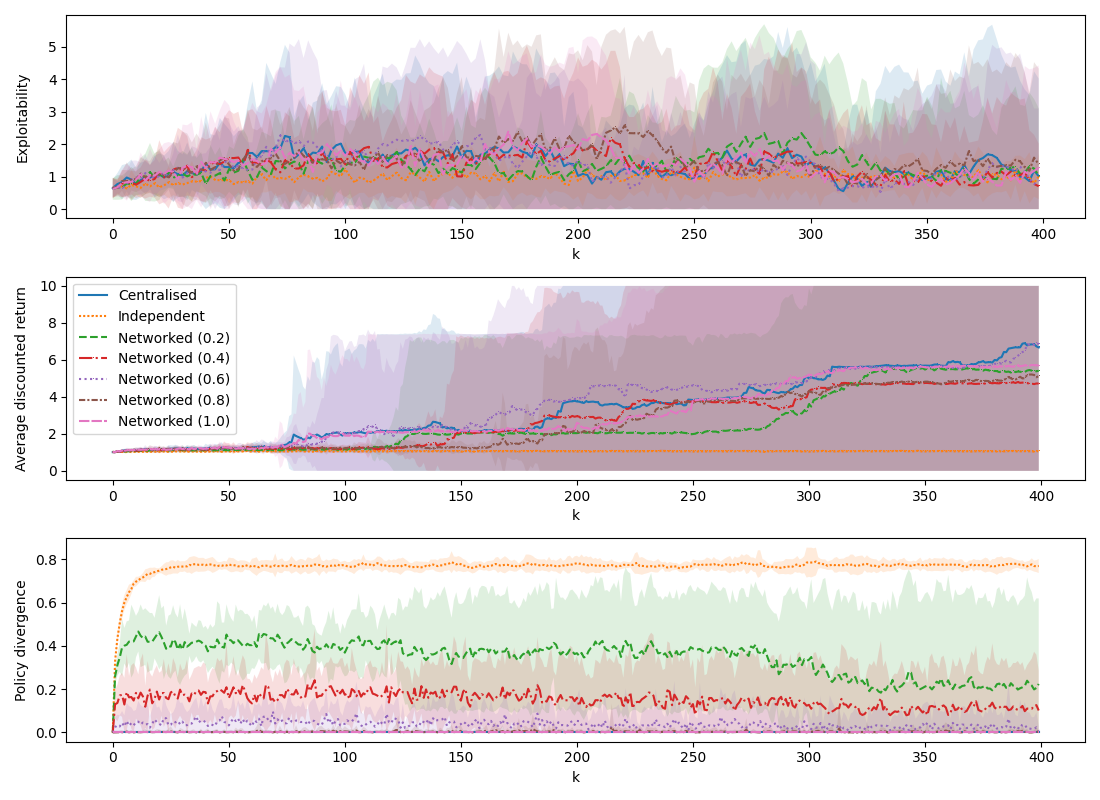}\caption{`Cluster' game on the larger 16x16 grid. While the independent-learning case has similar exploitability to the other settings, we can see that it is not actually learning to increase its return at all, making this an undesirable equilibrium. (I.e. agents are moving about randomly so there is little a deviating agent can do to increase its reward, hence exploitability is low even though the agents are not in fact clustered - see Rem. \ref{increasing_exploitability_remark}.) {All the networked settings perform similarly to the centralised case and outperform the return of the independent agents.}
  }
  \label{cluster_16}
\end{figure}

\begin{figure}[t]
  \centering
  \includegraphics[width=1.0\linewidth]{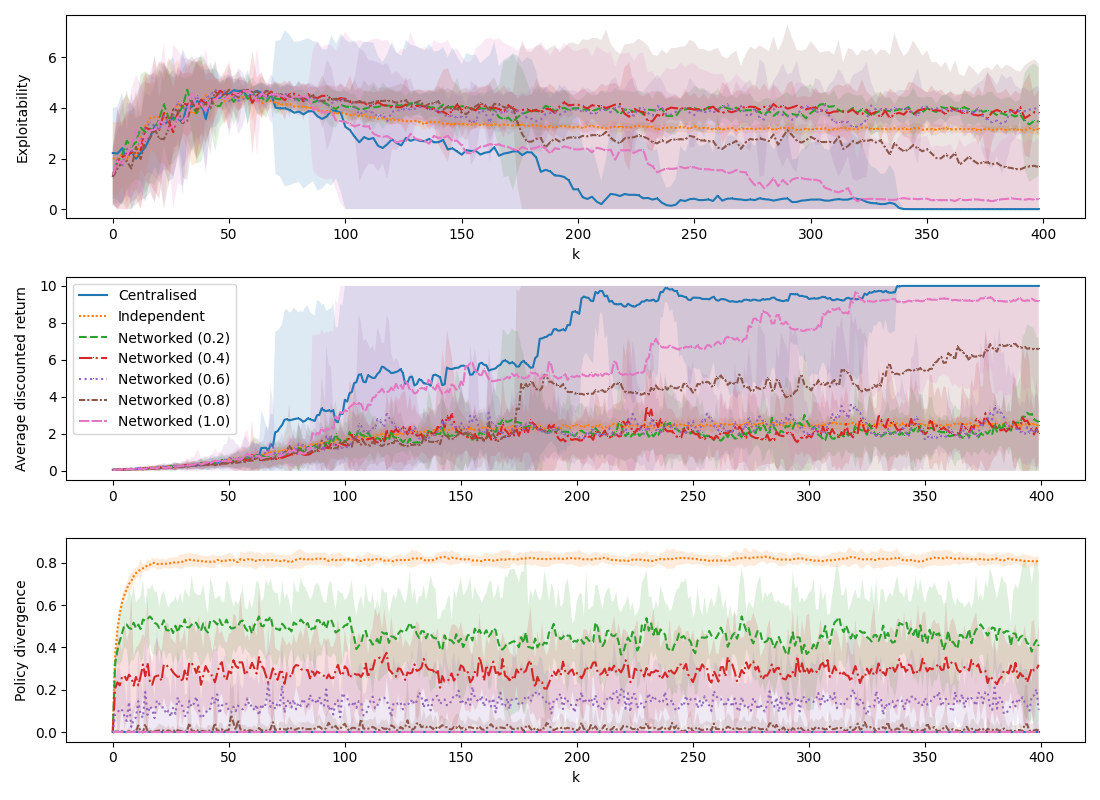}\caption{`Target agreement' game on the larger 16x16 grid. There is greater differentiation in this setting than in the 8x8 grid (Fig. \ref{agree_regular}) between the different broadcast radii in the networked cases, as might be expected in a less densely populated environment. The two largest broadcast radii (1.0, pink, and 0.8, brown), which have the most connected networks, outperform the independent case in terms of both exploitability and return. However, the other broadcast radii perform similarly to the independent case. 
  }
  \label{agree_16}
\end{figure}

\subsubsection{Ablation study of softmax temperature annealing scheme}

\begin{figure}[t]
  \centering
  \includegraphics[width=1.0\linewidth]{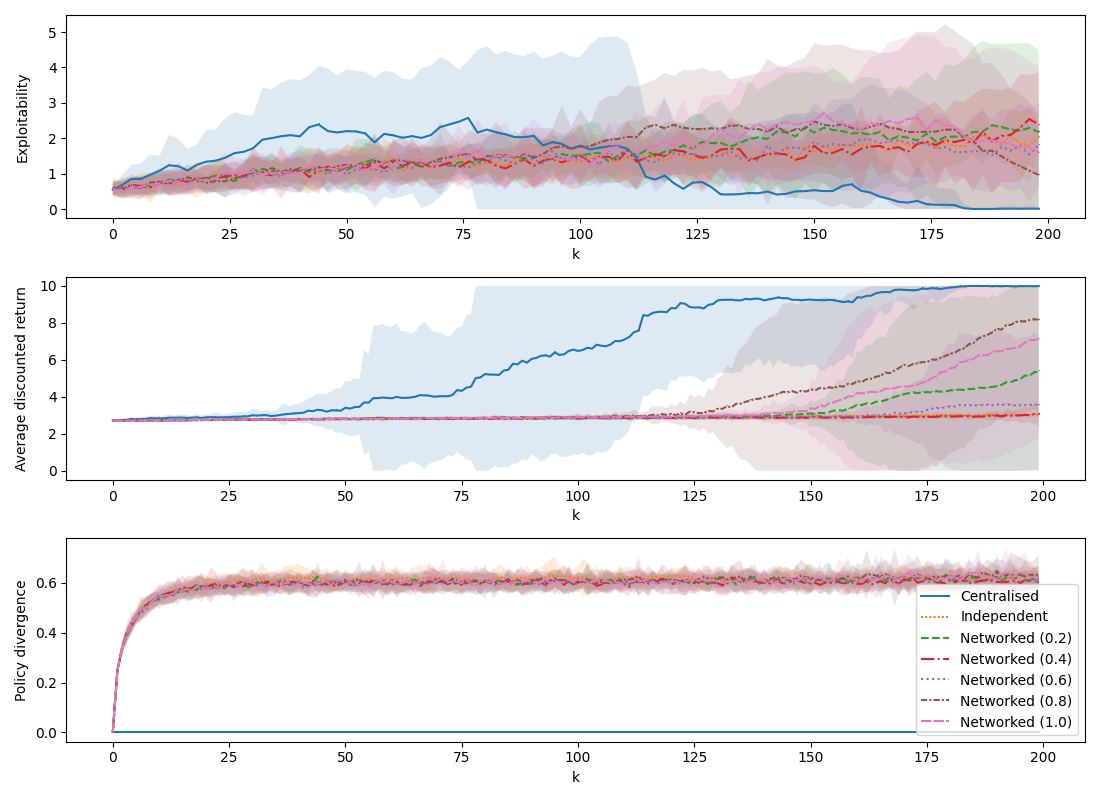}\caption{`Cluster' game with $\tau_k$ fixed as 100 for all $k$; compare this to Fig. \ref{cluster_regular} where $\tau_k$ is annealed. Without the annealing scheme, the networked architecture appears to perform similarly to the independent case in terms of exploitability, but {several broadcast radii outperform the independent case in terms of return, demonstrating that our networked algorithm can still help agents find `preferable’ equilibria}. However, whereas with annealing the networked architecture converges similarly to the centralised case, here it performs less well. 
  }
  \label{cluster_fixed}
\end{figure}

\begin{figure}[t]
  \centering
  \includegraphics[width=1.0\linewidth]{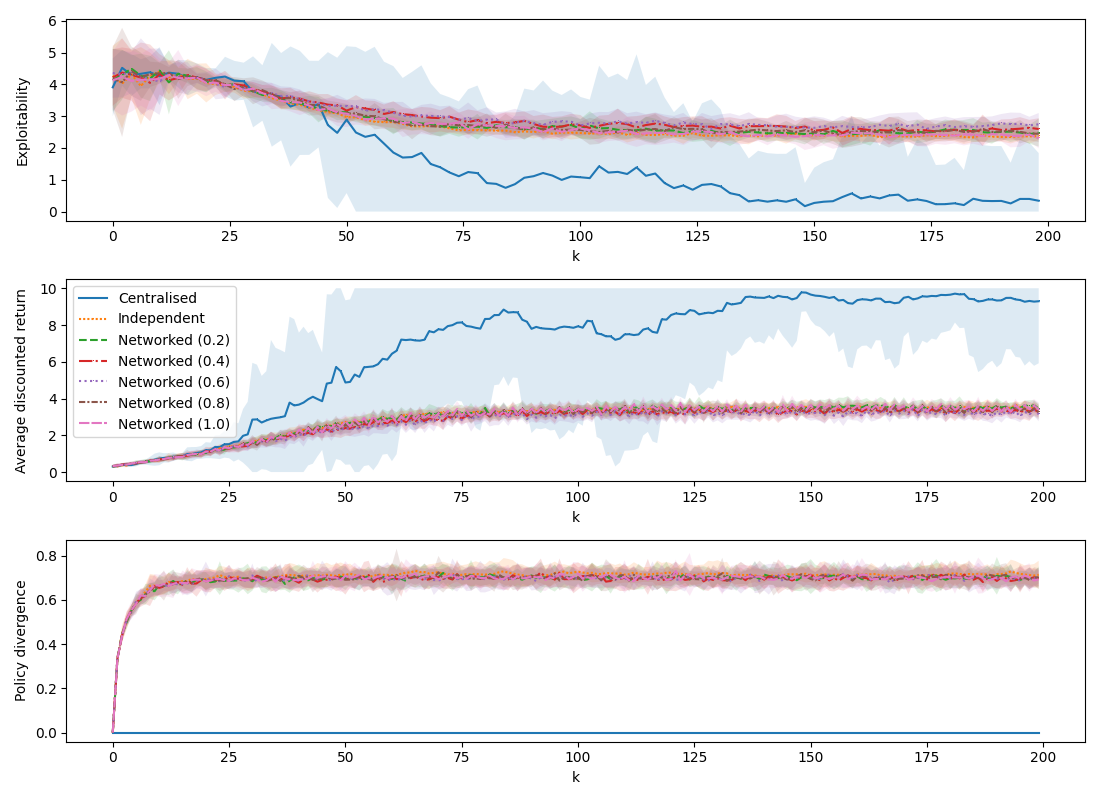}\caption{`Target agreement' game with $\tau_k$ fixed as 100 for all $k$. Without our annealing scheme for the softmax temperature, the networked architecture does not outperform the independent case. Compare this to Fig. \ref{agree_regular} which shows the benefit of annealing $\tau_k$. 
  }
  \label{agree_fixed}
\end{figure}

Figs. \ref{cluster_fixed} and \ref{agree_fixed} illustrate the effect of fixed $\{\tau_k\}_{k \in \{0,\dots,K-1\}} = 100$, where the networked architecture does not perform as well as if we use the stepped annealing scheme employed in all the other experiments and detailed in Table \ref{Hyperparameters}. The intuition behind the better performance achieved with the annealing scheme is as follows. If we begin with small $\tau_k$ (such that the softmax approaches being a max function), we heavily favour the adoption of the highest rewarded policies to speed up progress in the early stages of learning. Subsequently we increase $\tau_k$ in steps, promoting greater randomness in adoption, so that as the agents come closer to equilibrium, poorer policy updates that nevertheless receive a high return (due to randomness) do not introduce too much instability to learning and prevent convergence.

\section{Conclusion and future work}\label{future_work}

We contributed networked communication as a novel framework for learning MFGs from the empirical distribution, and provided accompanying theoretical and practical algorithms. We showed theoretically and experimentally that networked agents can considerably outperform independent ones, often performing similarly to the central-agent architecture while avoiding the restrictive assumption of the latter and its single point of failure.

Our experiments are based on relatively simple examples that demonstrate the advantages of our new approach, but which lack the complexity of the real-world applications to which we wish to address the approach. Moreover in our current experiments only the reward function depends on the mean-field distribution, and not the transition function, even though this is possible in theory; we will explore this element in future experiments. It is feasible that in more complex problems, it may not be possible to reduce hyperparameter values to the same extent we have demonstrated in our experimental examples. 

Moreover, real-world examples would likely require handling larger and continuous state/action spaces (the latter perhaps building on related work such as \citet{10448058}), which in turn may require (non-linear) function approximation. Future work therefore involves incorporating neural networks into our networked communication architecture for oracle-free, non-episodic MFG settings. Extending our algorithms in this way, which would depend on modifying the PMA step \citep{NEURIPS2020_2c6a0bae, wu2024populationaware}, would allow us to introduce communication networks to MFGs with \textit{non-stationary} equilibria, in addition to those with larger state/action spaces. Our method for non-stationary games will likely have agents' policies depending both on their local state and also on the population distribution \citep{9304340, survey_learningMFGs,perrin2022generalization,carmona2021modelfree}, but such a high-dimensional observation object is only possible with function approximation. The present work demonstrates the benefits of the networked communication architecture when the Q-function is poorly estimated and introduces experience relay buffers to the setting of learning from a non-episodic run of the empirical system. Both elements are an important bridge to employing (non-linear) function approximation in this setting, where the problems of data efficiency and imprecise value estimation can be even more acute, and where we may want to employ experience replay buffers to provide uncorrelated data to train the neural networks \citep{zhang2017deeper}. When the policy functions are approximated rather than tabular, our agents would communicate the functions' parameters instead of the whole policy as now.

In our future work with non-stationary equilibria, where agents' policies will also depend on the population distribution, it may be a strong assumption to suppose that decentralised agents with local state observations and limited communication radius would be able to observe the entire population distribution. We will therefore explore a framework of networked agents estimating the empirical distribution from only their local neighbourhood as in \citep{pomfrl}, and possibly also improving this estimation by communicating with neighbours \citep{subjective_equilibria}, such that this useful information spreads through the network along with policy parameters. 

Our algorithm for the networked case (Alg. \ref{networked_algorithm}), as well as the central-agent and independent baselines from \citet{policy_mirror_independent}, all have multiple nested loops. This is a potential limitation for real-world implementation, since the decentralised agents might be sensitive to failures in synchronising these loops. However, in practice, we show that our networked architecture provides redundancy and robustness (which the independent-learning algorithm lacks) in case of learning failures that may result from the necessities of synchronisation (see Sec. \ref{robustness_experiments}). We have also shown that networked communication in combination with the replay buffer allows us to reduce the hyperparameter $M_{td}$ to 1, essentially removing the inner `waiting' loop. Nevertheless, our algorithm still features multiple loops, and future work lies in simplifying the algorithms further to aid practical implementation, possibly by techniques such as asynchronous communication \citep{10490266}. Future works should also consider updating our theoretical guarantees in light of our current practical algorithmic enhancements, as well as any future modifications. 

Since the MFG setting is technically non-cooperative, we have pre-empted objections that agents would not have incentive to communicate their policies by focusing on coordination games, i.e. where agents seek to maximise only their individual returns, but receive higher rewards when they follow the same strategy as other agents. In this case they stand to benefit by exchanging their policies with others. Future work lies in extending our networked communication algorithms to mean-field control, the cooperative counterpart to MFGs, where agents would have incentive to communicate across different types of game. Nevertheless, in real-world settings, the communication network could still be vulnerable to malfunctioning agents or adversarial actors poisoning the equilibrium by broadcasting untrue policy information \citep{agrawal2024impact}, or equally to unreliable communication channels. It is outside the scope of this paper to analyse how much false information would have to be broadcast by how many agents to affect the equilibrium, but real-world applications may need to compute this and prevent it. Future research to mitigate this risk might build on work such as \citet{piazza2024power}, where `power regularisation' of information flow is proposed to limit the adverse effects of communication by misaligned agents.


While our MFG \textit{algorithms} are designed to handle arbitrarily large numbers of agents (and theoretically perform better as $N \rightarrow \infty$), the \textit{code} for our experiments naturally still suffers from a bottleneck of computational speed when simulating agents that in the real world would be acting and learning in parallel, since the GPU can only process JAX-vectorised elements in batches of a certain size.

\subsubsection*{Broader Impact Statement}
We identified no specific ethical concerns regarding our work, which explores new game theoretical and machine learning algorithms in general settings.

\bibliography{main}
\bibliographystyle{tmlr}

\end{document}